\documentclass[conference]{IEEEtran}
\usepackage{cite}
\usepackage{amsmath,amssymb,amsfonts}
\usepackage{amsthm}
\usepackage{algorithmic}
\usepackage{graphicx}
\usepackage{textcomp}
\usepackage{xcolor}
\def\BibTeX{{\rm B\kern-.05em{\sc i\kern-.025em b}\kern-.08em
    T\kern-.1667em\lower.7ex\hbox{E}\kern-.125emX}}

\usepackage[linesnumbered,ruled,vlined]{algorithm2e}
\usepackage[unicode=true,psdextra,hidelinks]{hyperref}
\usepackage{cleveref}
\usepackage{booktabs} 
\DeclareUnicodeCharacter{2061}{}

\usepackage{tikz}
\usetikzlibrary{shapes.geometric, shapes.multipart, positioning, arrows.meta}

\usepackage{orcidlink}

\newtheorem{definition}{Definition}
\newtheorem{exmp}{Example}[section]
\newtheorem{theorem}{Theorem}

\newtheorem{property}{Property}
\newtheorem{corollary}{Corollary}
\newtheorem{lemma}{Lemma}

\allowdisplaybreaks[4]

\begin{document}

\title{Bayesian Advantage of Re-Identification Attack in the Shuffle Model}

\author{\IEEEauthorblockN{Pengcheng Su}
\IEEEauthorblockA{School of Computer Science\\
Peking University\\ pcs@pku.edu.cn}
\and
\IEEEauthorblockN{Haibo Cheng\quad Ping Wang}
\IEEEauthorblockA{National Engineering Research Center of Software Engineering\\
Peking University\\ hbcheng@pku.edu.cn, pwang@pku.edu.cn}
}

\maketitle

\begin{abstract}
The shuffle model, which anonymizes data by randomly permuting user messages, has been widely adopted in both cryptography and differential privacy. In this work, we present the first systematic study of the Bayesian advantage in re-identifying a user's message under the shuffle model. We begin with a basic setting: one sample is drawn from a distribution $P$, and $n - 1$ samples are drawn from a distribution $Q$, after which all $n$ samples are randomly shuffled. We define $\beta_n(P, Q)$ as the success probability of a Bayes-optimal adversary in identifying the sample from $P$, and define the additive and multiplicative Bayesian advantages as $\mathsf{Adv}_n^{+}(P, Q) = \beta_n(P,Q) - \frac{1}{n}$ and $\mathsf{Adv}_n^{\times}(P, Q) = n \cdot \beta_n(P,Q)$, respectively.

We derive exact analytical expressions and asymptotic characterizations of $\beta_n(P, Q)$, along with evaluations in several representative scenarios. Furthermore, we establish (nearly) tight mutual bounds between the additive Bayesian advantage and the total variation distance.

Finally, we extend our analysis beyond the basic setting and present, for the first time, an upper bound on the success probability of Bayesian attacks in shuffle differential privacy. Specifically, when the outputs of $n$ users—each processed through an $\varepsilon$-differentially private local randomizer—are shuffled, the probability that an attacker successfully re-identifies any target user's message is at most $e^{\varepsilon}/n$.
\end{abstract}


\section{Introduction}

The shuffle model involves a trusted shuffler that receives messages from users, randomly permutes them, and then publishes the shuffled outputs in order to achieve anonymization \cite{Prochlo,cheu19}. This model has demonstrated strong capabilities in both cryptography and differential privacy: it enables information-theoretically secure key exchange (which can be further applied to advanced tasks such as secure multiparty computation) \cite{Ishai06,Beimei20}, and significantly amplifies the privacy guarantees of local differential privacy (LDP) mechanisms \cite{Erlingsson19,cheu19,Balle2019,Feldman2021}.

In this paper, we initiate the first systematic study of a fundamental question in the shuffle model: suppose $n$ users each independently generate a message $y_i$ according to a probability distribution $P_i$ for $i = 1, 2, \dots, n$. After applying the shuffler, the adversary observes the shuffled output $\boldsymbol{z} = \mathcal{S}(y_1, y_2, \dots, y_n) = (y_{\sigma(1)}, y_{\sigma(2)}, \dots, y_{\sigma(n)})$, where $\sigma$ is a random permutation\footnote{In this notation, $\sigma(i) = j$ means that the output originally from the $j$-th user is now placed in the $i$-th position.
 For example, if $\boldsymbol{z} = (y_2, y_4, \dots)$, then $\sigma(1) = 2, \ \sigma(2) = 4, \ \dots$
}. We ask: what is the probability that an adversary can correctly identify the message sent by a target user $t$—that is, recover the index $\sigma^{-1}(t)$?

The success probability of such an attack depends on the adversary’s prior knowledge. We assume a strong adversary who knows all distributions $P_i$. Under this assumption, the theoretically optimal strategy is to compute the posterior distribution via Bayes’ rule and select the most likely candidate \cite{1}.

We illustrate the notion of the re-identification attack in the shuffle model with a simple example.
Consider an anonymous voting scenario: voter~1's vote is known a priori to be either candidate~A or candidate~B, whereas every other voter is known {\em not} to vote for A or B.
Even after the ballots are randomly permuted by the shuffler, an adversary can immediately locate voter~1's ballot in the shuffled ballots with probability 1, because its observable value is unique.
This example is trivial in the sense that the support of voter~1's output distribution is disjoint from the supports of the other voters' distributions.

More generally, if the other voters' output distributions have nonempty overlap with voter~1's distribution, the observations become ambiguous and the adversary faces nontrivial uncertainty.
Analyzing this nontrivial regime is the main objective of this paper.

We begin by analyzing a basic setting in which the input distributions are specified as $P_1 = P$ and $P_2 = P_3 = \dots = P_n = Q$. This captures the scenario where a single sample drawn from $P$ is hidden among $(n - 1)$ independently drawn samples from $Q$, and the adversary's objective is to identify the one originating from $P$.

This setting is not only analytically tractable but also practically motivated: it forms the basis of so-called \emph{honey techniques}, with the \emph{Honeyword} system being a representative example (see Section~2 for details)~\cite{1,9,10}. In contrast, the general case—where the distributions $P_2, \dots, P_n$ are not necessarily identical—is typically analytically intractable \cite{10221886}. Fortunately, we show that it can be reduced to this basic setting without loss of generality.

\begin{figure}[t]
\centering
\resizebox{\linewidth}{!}{%
\begin{tikzpicture}[
  font=\large, 
  user/.style={rectangle, draw, rounded corners, minimum width=1.5cm, minimum height=0.6cm},
  rnd/.style={ellipse, draw, minimum width=1.2cm, minimum height=0.6cm},
  shuffler/.style={rectangle, draw, fill=gray!10, minimum width=2.5cm, minimum height=1cm},
  msg/.style={rectangle, draw, fill=blue!10, minimum width=0.8cm, minimum height=0.5cm},
  adv/.style={rectangle, draw, thick, fill=red!10, minimum width=1.6cm, minimum height=0.8cm}
]

\node[user] (U1) at (0, 0) {User 1};
\node[user] (U2) [below=0.6cm of U1] {User 2};
\node[user] (U3) [below=0.5cm of U2] {$\vdots$};
\node[user] (Un) [below=0.6cm of U3] {User $n$};

\node[rnd] (R1) [right=1.5cm of U1] {$P_1$};
\node[rnd] (R2) [right=1.5cm of U2] {$P_2$};
\node[rnd] (Rn) [right=1.5cm of Un] {$P_n$};

\draw[->] (U1) -- (R1);
\draw[->] (U2) -- (R2);
\draw[->] (Un) -- (Rn);

\node[msg] (Y1) [right=1.2cm of R1] {$y_1$};
\node[msg] (Y2) [right=1.2cm of R2] {$y_2$};
\node[msg] (Yn) [right=1.2cm of Rn] {$y_n$};

\draw[->] (R1) -- (Y1);
\draw[->] (R2) -- (Y2);
\draw[->] (Rn) -- (Yn);

\node[shuffler] (S) [right=1.8cm of Y2, yshift=-0.4cm] {Shuffler $\mathcal{S}$};

\draw[->] (Y1) -- ([yshift=0.2cm]S.west);
\draw[->] (Y2) -- (S.west);
\draw[->] (Yn) -- ([yshift=-0.2cm]S.west);

\node[msg] (Z1) [right=2.5cm of S.north east, yshift=1.0cm] {$z_1 = y_{\sigma(1)}$};
\node[msg] (Z2) [below=0.5cm of Z1] {$z_2 = y_{\sigma(2)}$};
\node[msg] (Z3) [below=0.5cm of Z2] {$\vdots$};
\node[msg] (Zn) [below=0.5cm of Z3] {$z_n = y_{\sigma(n)}$};

\draw[->] (S.east) -- (Z1.west);
\draw[->] (S.east) -- (Z2.west);
\draw[->] (S.east) -- (Zn.west);

\node[adv] (A) [right=1.5cm of Z2] {Adversary};

\draw[->] (Z1.east) -- (A.west);
\draw[->] (Z2.east) -- (A.west);
\draw[->] (Zn.east) -- (A.west);

\draw[->, thick, red] ([yshift=-1pt]A.south) -- ++(0,-0.6) node[below] {\footnotesize Output $\sigma^{-1}(t)$};

\draw[<- , thick, red] ([yshift=1pt]A.north) -- ++(0,0.6) node[above] {\footnotesize Input $t$};

\end{tikzpicture}
}
\caption{Re-identification attack in the shuffle model. Each user $i$ generates $y_i \sim P_i$, and the shuffler randomly permutes the outputs. The adversary observes $\boldsymbol{z} = (y_{\sigma(1)}, \dots, y_{\sigma(n)})$ and tries to infer the position of a target message.}
\label{fig:shuffle_attack}
\end{figure}
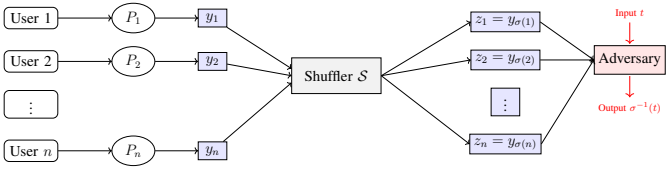

We define the success probability of the Bayes-optimal adversary in the basic setting as $\beta_n(P, Q)$. Since a trivial strategy that outputs a uniformly random index in $[n] = \{1, 2, \dots, n\}$ achieves a success probability of $1/n$, we define the \emph{Bayesian advantage} in both additive and multiplicative forms as follows:
\begin{align*}
\mathsf{Adv}_n^{+}(P, Q) &= \beta_n(P, Q) - \frac{1}{n},\\
\mathsf{Adv}_n^{\times}(P, Q) &= \frac{\beta_n(P, Q)}{1/n} = n \cdot \beta_n(P, Q).
\end{align*}

Leveraging the classical likelihood ratio method from information theory~\cite{Polyanskiy22}, we derive an exact analytical expression for $\beta_n(P, Q)$. This expression enables evaluation of $\beta_n(P, Q)$ in several representative cases and facilitates a precise characterization of its asymptotic behavior as $n$ increases:
\[
\beta_n(P, Q) = \alpha + \frac{M}{n} - \ell(n), \ \text{where } \ell(n) \ge 0 \text{ and } \ell(n) = o\left(\frac{1}{n}\right),
\]
with $M = \sup_{x: Q(x) \neq 0} \frac{P(x)}{Q(x)}$ and $\alpha = \sum_{x: Q(x) = 0} P(x)$.
This result reveals that when $P$ is absolutely continuous with respect to $Q$ (i.e., $\alpha = 0$), the success probability is tightly upper bounded by $\frac{M}{n}$.

Furthermore, we establish a connection between $\mathsf{Adv}_n^{+}(P, Q)$ and the classical total variation distance from information theory, defined as $\Delta(P, Q) = \frac{1}{2} \sum_{x} |P(x) - Q(x)|$. Specifically, we prove the following tight two-sided bound:
\[
\frac{\Delta(P, Q)}{n} \le \mathsf{Adv}_n^{+}(P, Q) \le \Delta(P, Q).
\]
This result implies that for any fixed $n$, the additive Bayesian advantage $\mathsf{Adv}_n^{+}(P, Q)$ is equivalent to $\Delta(P, Q)$ up to constant factors. In particular, as $\Delta(P, Q) \to 0$, the additive advantage $\mathsf{Adv}_n^{+}(P, Q) \to 0$.
We further demonstrate that both bounds are tight by providing explicit examples. The cryptographic implications of this connection are discussed in Section~\ref{sec:connections}.

Finally, we extend the analysis from the basic setting to the general setting. As a representative application, we consider the \emph{single-message shuffle model}, a widely adopted paradigm in differential privacy \cite{Erlingsson19,Balle2019,cheu19,Feldman2021}. In this model, each user holds an input $x_i \in \mathbb{X}$ and applies a local randomizer $\mathcal{R}$ to produce a message $\mathcal{R}(x_i)$ \footnote{{
It is worth noting that our analysis allows the inputs $x_i$ to be random variables on $\mathbb{X}$. A detailed discussion is provided in Section \ref{sec_definition_attack_dp}.
}}. These messages are then passed to a trusted shuffler, which applies a uniformly random permutation before forwarding the shuffled messages to the server.

Previous research has primarily focused on analyzing and computing the differential privacy guarantees in the shuffle model (see details in Section \ref{sec_background_dp}) \cite{Erlingsson19,Balle2019,cheu19,Feldman2021,feldman2023soda}. To the best of our knowledge, our work is the first to formally consider and theoretically analyze re-identification attacks in the Shuffle DP setting.

We prove that if the mechanism $\mathcal{R}$ satisfies $\varepsilon$-differential privacy—i.e., for all $x, x' \in \mathbb{X}$ and all outputs $y\in \mathbb{Y}$,
\[
\frac{\Pr[\mathcal{R}(x) = y]}{\Pr[\mathcal{R}(x') = y]} \le e^{\varepsilon},
\]
then the success probability of a Bayesian adversary in re-identifying a user's output is bounded by
\[
\beta_n(\mathcal{R}) \le \frac{e^{\varepsilon}}{n}.
\]
This result indicates that in the shuffle DP model, the adversary’s success probability in re-identifying any user's output is at most $e^{\varepsilon}$ times the ideal baseline of $1/n$.
In other words, the multiplicative Bayesian advantage $\mathsf{Adv}_n^{\times}(\mathcal{R}) \le e^{\varepsilon}$.
A more refined analysis shows that when user~1 is targeted, the tight asymptotic bound on $\mathsf{Adv}_n^{\times}(\mathcal{R},x_1)$ is given by
\[
\mathsf{Adv}_n^{\times}(\mathcal{R},x_1) \le M,
\]
where $M=\sup_{y} \frac{\Pr[\mathcal{R}(x_1) = y]}{\inf_{x} \Pr[\mathcal{R}(x) = y]}$.

Our main proof technique is based on a novel application of the two decomposition methods known as the \emph{clone} \cite{Feldman2021} and the \emph{blanket} \cite{Balle2019}. These decompositions were originally developed to analyze privacy-amplification effects in the shuffle model; the reader may consult \cite{su2025decompositionbasedoptimalboundsprivacy} for a survey of these methods. From a theoretical perspective, the reason decomposition techniques are well suited to analyzing both privacy amplification and re‑identification attacks is that \textit{both problems admit closed‑form solutions in the basic setting}. 
In particular, Lemma 5.3 in the blanket paper \cite{Balle2019} provides an analytic expression for the Hockey‑Stick divergence of the shuffled outputs when user~1’s input changes while all other users’ messages are sampled i.i.d.\ from the same distribution; our basic‑setting analysis for re‑identification (Theorem \ref{theorem2}) is analogous.

Intuitively, when the distributions of the other users are not identical, the decomposition approach discards the heterogeneous parts of those distributions and retains only the component that is common (i.e., the overlapping distribution). One can show that discarding the non‑common parts can only increase the adversary’s distinguishing power. Hence, applying these decompositions yields valid upper bounds on the shuffle mechanism’s differential‑privacy loss and on the Bayesian optimal re‑identification success probability.

It is worth mentioning that we find the blanket decomposition to be optimal among all possible decompositions in the analysis of re-identification attacks. This parallels the result in~\cite{su2025decompositionbasedoptimalboundsprivacy}, which established the optimality of the blanket decomposition for privacy amplification analysis. Overall, the clone decomposition, due to its simple form, provides consistent and concise bounds, whereas the blanket decomposition can yield tighter, mechanism-specific tailored bounds for particular differential privacy mechanisms.



\textbf{Our contributions are summarized as follows:}
\begin{itemize}
  \item We introduce a formal definition of the Bayesian success probability $\beta_n(P, Q)$ for identifying a sample drawn from $P$ among $n-1$ decoys drawn from $Q$ under random shuffling, and define the corresponding Bayesian advantage in both additive and multiplicative forms.
  
  \item We derive an exact analytical expression for $\beta_n(P, Q)$ using the likelihood ratio method, and characterize its asymptotic behavior.
  
  \item We establish mutual bounds connecting the Bayesian advantage and the total variation distance,
  showing that they are equivalent up to constant factors.
  
  \item We extend our analysis to the general shuffle model and prove that if the local randomizer $\mathcal{R}$ satisfies $\varepsilon$-DP, then the Bayesian re-identification success probability satisfies $\beta_n(\mathcal{R}) \le \frac{e^{\varepsilon}}{n}$. We also derive a tight asymptotic bound on the multiplicative advantage.
\end{itemize}

\section{Background}

\subsection{Differential privacy}\label{sec_background_dp}

Differential privacy is a privacy-preserving framework for randomized algorithms. Intuitively, an algorithm is differentially private if the output distribution does not change significantly when a single individual's data is modified. This ensures that the output does not reveal substantial information about any individual in the dataset.

The classical definition of differential privacy (DP) assumes that the server has direct access to the data of all $n$ users \cite{Dwork2006}. In real-world applications, \emph{Local Differential Privacy} (LDP) is widely adopted, as it eliminates the need for a trusted curator by applying randomized noise to each user's data before any aggregation takes place~\cite{Cormode18,LDPsurvey,Microsoft17,Apple17}. However, this decentralized approach often incurs substantial utility loss due to the high level of noise required to ensure privacy.

To address this trade-off, the \textit{shuffle model} introduces a trusted shuffler between the users and the aggregator~\cite{cheu19,Erlingsson19,Prochlo}. In this work, we focus on the \emph{single-message shuffle model}, in which each client sends a single report generated from their data, and these reports are anonymized by random shuffling before being forwarded to the server~\cite{feldman2023soda}. Here we give the formal definitions of DP, LDP, and shuffle DP.

We say that random variables \( P \) and \( Q \) are $(\varepsilon, \delta)$-indistinguishable if for all set $T$:
\[
\Pr[P\in T]\le e^{\varepsilon}\Pr[Q\in T]+\delta.
\]
If two datasets \( X \) and \( X' \) have the same size and differ only by the data of a single individual, they are referred to as neighboring datasets (denoted by \( X \simeq X' \)).

\begin{definition}[Differential Privacy]
An algorithm \( \mathcal{R} : \mathbb{X}^n \to \mathbb{Z} \) satisfies $(\varepsilon, \delta)$-differential privacy if for all neighboring datasets \( X, X' \in \mathbb{X}^n \), \( \mathcal{R}(X) \) and \( \mathcal{R}(X') \) are $(\varepsilon, \delta)$-indistinguishable.    
\end{definition}

\begin{definition}[Local Differential Privacy]
An algorithm \( \mathcal{R} : \mathbb{X} \to \mathbb{Y} \) satisfies local $(\varepsilon, \delta)$-differential privacy if for all \( x, x' \in \mathbb{X} \), \( \mathcal{R}(x) \) and \( \mathcal{R}(x') \) are $(\varepsilon, \delta)$-indistinguishable.    
\end{definition}

Here, \( \varepsilon \) is referred to as the privacy budget, which controls the privacy loss, while \( \delta \) allows for a small probability of failure. When \( \delta = 0 \), the mechanism is also called \( \varepsilon \)-DP.

A single-message protocol \( \mathcal{P} \) in the shuffle model is defined as a pair of algorithms \( \mathcal{P} = (\mathcal{R}, \mathcal{A}) \), where \( \mathcal{R} : \mathbb{X} \to \mathbb{Y} \), and \( \mathcal{A} : \mathbb{Y}^n \to \mathbb{O} \). We call \( \mathcal{R} \) the \textit{local randomizer}, \( \mathbb{Y} \) the \textit{message space} of the protocol, \( \mathcal{A} \) the \textit{analyzer}, and \( \mathbb{O} \) the \textit{output space}~\cite{Balle2019,cheu19}. 

The overall protocol implements a mechanism \( \mathcal{P} : \mathbb{X}^n \to \mathbb{O} \) as follows: Each user \( i \) holds a data record \( x_i \), to which they apply the local randomizer to obtain a message \( y_i = \mathcal{R}(x_i) \). The messages \( y_i \) are then shuffled and submitted to the analyzer. Let \( \mathcal{S}(y_1, \dots, y_n) \) denote the random shuffling step, where \( \mathcal{S} : \mathbb{Y}^n \to \mathbb{Y}^n \) is a \textit{shuffler} that applies a random permutation to its inputs. 

In summary, the output of \( \mathcal{P}(x_1, \dots, x_n) \) is given by
\[
\mathcal{A} \circ \mathcal{S} \circ \mathcal{R}^n (\boldsymbol{x}) = \mathcal{A}(\mathcal{S}(\mathcal{R}(x_1), \dots, \mathcal{R}(x_n))).
\]

\begin{definition}[Differential Privacy in the Shuffle Model]
A protocol \( \mathcal{P} = (\mathcal{R}, \mathcal{A}) \) satisfies $(\varepsilon, \delta)$-differential privacy in the shuffle model if for all neighboring datasets \( X, X' \in \mathbb{X}^n \), the distributions \( \mathcal{S} \circ \mathcal{R}^n(X) \) and \( \mathcal{S} \circ \mathcal{R}^n(X') \) are $(\varepsilon, \delta)$-indistinguishable.    
\end{definition}

The ``amplification-by-shuffling" theorem in the shuffle model implies that when each of the \(n\) users randomizes their data using an \(\varepsilon_0\)-LDP mechanism, the collection of shuffled reports satisfies \((\varepsilon(\varepsilon_0, \delta, n), \delta)\)-DP, where \(\varepsilon(\varepsilon_0, \delta, n) \ll \varepsilon_0\) for sufficiently large \(n\) and not too small $\delta$~\cite{feldman2023soda}.

Previous research has primarily focused on analyzing and computing the differential privacy guarantees—namely, the parameters $(\varepsilon, \delta)$—achieved by a given protocol $\mathcal{P}$ in the shuffle model \cite{Erlingsson19,Balle2019,cheu19,Feldman2021,feldman2023soda}. Our work is the first to formally consider and theoretically analyze re-identification attacks in the Shuffle DP setting.

\subsection{Honeyword}

Frequent password breaches—over 4{,}450 incidents exposing 28 billion records by July 2025~\cite{44}—have motivated leakage detection mechanisms such as the Honeyword system~\cite{9}.

In this system, the authentication server stores a user's real password alongside $n - 1$ plausible decoys—called \emph{honeywords}—randomly shuffled to conceal the real one~\cite{9}. Upon a breach, the attacker obtains $n$ indistinguishable candidates rather than the actual password. If a honeyword is later used in a login attempt, the system detects this anomaly and raises an alert~\cite{9,1}. 

Let $P$ denote the distribution of human-chosen passwords and $Q$ the distribution of honeywords. The security of the Honeyword system is then captured by the Bayesian success probability $\beta_n(P, Q)$.

Prior studies on honeyword security typically consider an adversary who has access to a dataset $D$ of human-chosen passwords, which is assumed to be sampled from the underlying distribution $P$, along with knowledge of the honeyword generation mechanism $Q$ \cite{1,21,23,40}. Due to the lack of full knowledge of $P$, the optimal attack strategy in this setting is analytically intractable. As a result, existing works rely on designing heuristic attack algorithms to evaluate the security of honeyword systems.

In contrast, we consider, for the first time, a theoretically strongest adversary who has complete knowledge of the distribution $P$. Under this setting, we provide an exact theoretical characterization of the Bayes-optimal success probability $\beta_n(P, Q)$, including its asymptotic analysis and its relationship to the total variation distance. We believe this perspective offers new insights into the security analysis of honeyword systems and related mechanisms.

\section{Problem formulation}
\subsection{Re-identification attack in the basic setting}\label{sec_definition_basic}
The formal definition of the re-identification attack in the basic setting is given in Game~\ref{guess_game}. Game~\ref{guess_game} additionally introduces a parameter $k$, allowing the adversary $\mathcal{A}$ to submit $k$ guesses. If the position of the sample drawn from distribution $P$ (i.e., $\sigma^{-1}(1)$) is among the $k$ guesses submitted by $\mathcal{A}$, then the adversary wins the game. We use the notation $y_1 \gets_{P} \mathbb{Y}$ to denote that $y_1$ is sampled from distribution $P$ over the domain $\mathbb{Y}$, and we write $\mathcal{P}_n$ to denote the set of all permutations over $[n]$.

\begin{definition}
We define the maximum success probability of such an adversary as
\[
\beta_n^k(P, Q) = \max_{\mathcal{A}} \Pr\left[\mathbf{GuessGame}_{\mathcal{A}}^{P,Q,n}(k) = \mathbf{True} \right].
\]
The special case $k = 1$, denoted $\beta_n^1(P, Q)$, corresponds to the single-guess setting, which is of primary interest in practice. For simplicity, we refer to it as $\beta_n(P, Q)$.
\end{definition}

We consider an information-theoretic adversary: that is, an attacker with full knowledge of the distributions $P$ and $Q$, and unbounded computational power. In Section~\ref{sec:bayesian_attack}, we show that the optimal attack strategy for such an adversary is to follow the posterior distribution induced by Bayes’ rule. Consequently, we refer to $\beta_n(P, Q)$ as the \emph{Bayesian success probability}.

\textbf{Comparison with alternative definitions.}
It is important to emphasize that in our formulation, the adversary is tasked with identifying the \emph{position} $\sigma^{-1}(1)$ of the sample drawn from $P$, rather than recovering its actual \emph{value} $y_1$. This distinction is essential, as defining success based on recovering the value $y_1$ leads to metrics that are not robust.

Consider, for instance, the trivial case where both $P$ and $Q$ deterministically output the same value (e.g., $1$). In this scenario, an adversary that simply outputs $y_1 = 1$ will always succeed, achieving a success probability of $1$.

As another example, let $P$ and $Q$ be uniform distributions over $[m]$. As $m$ grows, the probability of value collisions in the multiset $\{y_1, \dots, y_n\}$ vanishes. Without such collisions, attempting to recover the position of the sample drawn from $P$ reduces to uniform random guessing, yielding success probability exactly $1/n$.

By contrast, as formally shown in Section~4.1, when $P = Q$, the success probability of any adversary attempting to identify $\sigma^{-1}(1)$ is exactly $1/n$, \emph{regardless of the structure of $P$ and $Q$}. Furthermore, analyzing the success probability of recovering the actual value $y_1$ is often \emph{analytically intractable}~\cite{10221886}. In comparison, $\beta_n(P, Q)$ admits a closed-form expression and is computationally tractable.

\renewcommand{\algorithmcfname}{Game}
\begin{algorithm}[t]
$y_1 \gets_{P} \mathbb{Y}$\\
$y_2,y_3,\dots,y_n \gets_{Q} \mathbb{Y}$\\
$\boldsymbol{y}\gets(y_1,y_2,y_3,\dots,y_n)$\\
$\sigma \gets_{\$} \mathcal{P}_n$ \tcp{Sample a random permutation}
$\boldsymbol{y}_{\sigma}\gets (y_{\sigma(1)},y_{\sigma(2)},\dots,y_{\sigma(n)})$\\
$G=\{g_1,g_2,\cdots,g_k\} \gets \mathcal{A}(\boldsymbol{y}_{\sigma})$\\
\textbf{return} $\sigma^{-1}(1)\in G$.
\caption{$\mathbf{GuessGame}^{P,Q,n}_{\mathcal{A}}(k)$}\label{guess_game}

\end{algorithm}

\renewcommand{\algorithmcfname}{Game}
\begin{algorithm}[t]
\For{$i=1\ \mathrm{to}\ n$}{$y_i \gets_{\mathcal{R}(x_i)} \mathbb{Y}$}
$\boldsymbol{y}\gets (y_1,y_2,y_3,\dots,y_n)$\\
$\sigma \gets_{\$} \mathcal{P}_n$ \tcp{Sample a random permutation}
$\boldsymbol{y}_{\sigma}\gets (y_{\sigma(1)},y_{\sigma(2)},\dots,y_{\sigma(n)})$\\
$g \gets \mathcal{A}(\boldsymbol{y}_{\sigma})$\\
\textbf{return} $g=\sigma^{-1}(1)$.
\caption{$\mathbf{GuessGame}^{\mathcal{R},n}_{\mathcal{A}}(\boldsymbol{x})$}\label{guess_game2}

\end{algorithm}

\subsection{Re-identification attack in Shuffle DP}\label{sec_definition_attack_dp}
The formal definition of the re-identification attack in the shuffle DP setting is given in Game~\ref{guess_game2}. For simplicity, we focus on the single-guess case; however, our analysis naturally extends to settings where the adversary is allowed multiple guesses.

For clarity, we first present the definitions of the Bayesian attack success probability and its associated advantage when $x_1, x_2, \dots, x_n$ take deterministic values. We then explain how these definitions can be naturally generalized to the case where each $x_i$ is drawn from a probability distribution $V_i$.

\begin{definition}\label{definition_rsp_shuffle}
The \emph{Bayesian re-identification success probability} for a user with input $x_1$ in the Shuffle DP model is defined as
\[
\beta_n(\mathcal{R}, x_1) = \max_{\mathcal{A},\, \boldsymbol{x} \in \mathbb{X}^n,\, \boldsymbol{x}[1] = x_1} 
\Pr\left[\mathbf{GuessGame}^{\mathcal{R}, n}_{\mathcal{A}}(\boldsymbol{x}) = \mathbf{True}\right],
\]
where the maximum is taken over all adversaries $\mathcal{A}$ and all input vectors $\boldsymbol{x}$ such that the first user's input is fixed to $x_1$.
\end{definition}

\begin{definition}
Define the \emph{worst-case Bayesian success probability} in the Shuffle DP model as
\[
\beta_n(\mathcal{R}) := \max_{x_1 \in \mathbb{X}} \beta_n(\mathcal{R}, x_1).
\]
\end{definition}

\begin{definition}
We define the corresponding \emph{Bayesian advantage} in both additive and multiplicative forms:
\begin{align*}
\mathsf{Adv}_n^{+}(\mathcal{R},x_1) &= \beta_n(\mathcal{R},x_1) - \frac{1}{n},\\
\mathsf{Adv}_n^{\times}(\mathcal{R},x_1) &= \beta_n(\mathcal{R},x_1) / \frac{1}{n}=n\cdot\beta_n(\mathcal{R},x_1),\\
\mathsf{Adv}_n^{+}(\mathcal{R})&=\max_{x_1\in\mathbb{X}} \mathsf{Adv}_n^{+}(\mathcal{R},x_1),\\
\mathsf{Adv}_n^{\times}(\mathcal{R})&=\max_{x_1\in\mathbb{X}} \mathsf{Adv}_n^{\times}(\mathcal{R},x_1).
\end{align*}
\end{definition}

In our theoretical analysis, we consider a powerful adversary who has full knowledge of the distributions $\mathcal{R}(x_i),i=1,2,\dots,n$. This strong assumption allows our results to upper-bound the capabilities of weaker adversaries who may only have partial information.

\emph{Interpretation.} It is worth noting that the inputs $x_i$ may themselves be random variables over $\mathbb{X}$. For instance, we may have $x_i \sim V_i$, where $V_i$, $i=1,2,\dots,n$, are probability distributions on $\mathbb{X}$. In this setting, we assume that the adversary knows each $V_i$ and can additionally determine the corresponding output distribution $P_i = \mathcal{R}(x_i)$.

Let $\mathcal{D}[\mathbb{X}]$ denote the set of all probability distributions over $\mathbb{X}$.
In the following, we distinguish between the notions of a random variable and a probability distribution: a probability distribution is a deterministic object, which can be described by a vector of probabilities; a random variable, on the other hand, is inherently stochastic and is said to follow a certain probability distribution.

It can be shown that for any probability distributions $V_1, V_2 \in \mathcal{D}[\mathbb{X}]$, its corresponding random variables $v_1,v_2$ and any $\varepsilon$-DP mechanism $\mathcal{R}$, the distributions $\mathcal{R}(v_1)$ and $\mathcal{R}(v_2)$ also satisfy $\varepsilon$-differential privacy:
\[
\forall y: \frac{\Pr[\mathcal{R}(v_1)=y]}{\Pr[\mathcal{R}(v_2)=y]} 
\le \frac{\sup_{x\in \mathbb{X}} \Pr[\mathcal{R}(x)=y]}{\inf_{x\in \mathbb{X}} \Pr[\mathcal{R}(x)=y]} \le e^{\varepsilon},
\]
where we use $\Pr[\mathcal{R}(v_1)=y] = \sum_x \Pr[v_1=x] \cdot \Pr[\mathcal{R}(x)=y] \in \big[\inf_{x\in \mathbb{X}} \Pr[\mathcal{R}(x)=y], \sup_{x\in \mathbb{X}} \Pr[\mathcal{R}(x)=y]\big]$, and the result follows directly from the definition of $\varepsilon$-DP.

Consequently, in the scenario where $x_i \sim V_i$, one can equivalently consider a new $\varepsilon$-DP mechanism 
\[
\mathcal{R}^*: \mathcal{D}[\mathbb{X}] \to \mathbb{Y}, \quad \mathcal{R}^*(V_i) = \mathcal{R}(x_i).
\]
This perspective transforms the input from a random variable $x_i$ into a deterministic probability distribution description $V_i$.

We illustrate the above transformation using the $\ln(3)$-DP 2-ary randomized response mechanism \cite{Warner65}.
\begin{exmp}
Denote \( \{1,2,\dots,k\} \) by \( [k] \) and the uniform distribution on \( [k] \) by \( \mathcal{U}_{[k]} \). 
For any \( k \in \mathbb{N} \) and \( \varepsilon > 0 \), the \( k \)-randomized response mechanism \( k\text{RR}: [k] \to [k] \) is defined as:
\[
k\text{RR}(x) =
\begin{cases}
x, & \text{with probability } \frac{e^{\varepsilon_0} - 1}{e^{\varepsilon} + k - 1}, \\
y \sim \mathcal{U}_{[k]}, & \text{with probability } \frac{k}{e^{\varepsilon} + k - 1}.
\end{cases}
\]
If $x_1$ takes the fixed value $1$, it corresponds to $V_1 = (1.0, 0)$. If $x_1$ takes value $1$ with probability $0.4$ and value $2$ with probability $0.6$, then $V_1 = (0.4, 0.6)$. For $V=(p,1-p)$, the mechanism $\mathcal{R}^*(V)$ defines a probability distribution over $\{1,2\}$, with probabilities given by $(0.25 + 0.5p,\, 0.75 - 0.5p)$.
\end{exmp}

For notational convenience, we will henceforth write $\mathcal{R}(x_i)$; the reader may interpret this as $\mathcal{R}^*(V_i)$ in cases where $x_i \sim V_i$. 
For example, Definition \ref{definition_rsp_shuffle} can be interpreted as follows:
\begin{definition}\label{re_definition_rsp_shuffle}
The \emph{Bayesian re-identification success probability} for a user with input $x_1 \sim V_i$ in the Shuffle DP model is defined as
\begin{align*}
\beta_n&(\mathcal{R},x_1):=
\beta_n(\mathcal{R}^*, V_1)\\
&= \max_{\mathcal{A},\, \boldsymbol{V} \in \mathcal{D}[\mathbb{X}]^n,\, \boldsymbol{V}[1] = V_1} 
\Pr\left[\mathbf{GuessGame}^{\mathcal{R}^*, n}_{\mathcal{A}}(\boldsymbol{V}) = {\mathbf{True}}\right],
\end{align*}
where the maximum is taken over all adversaries $\mathcal{A}$ and all input \textbf{distribution} vectors $\boldsymbol{V}=(V_1,V_2,\dots,V_n)$ such that the first user's input \textbf{distribution} is fixed to $V_1$.
\end{definition}
Since $\mathcal{R}^*$ still satisfies $\epsilon$-DP, the clone decomposition remains valid.
In addition, the blanket decomposition of $\mathcal{R}^*$ coincides with that of $\mathcal{R}$ (see Section~\ref{sec_blanket} for details). This justifies that our analysis applies without loss of generality.

\subsection{Relation to quantitative information flow}
This work focuses on re-identification attacks in the shuffle model. The most closely related line of research is \emph{Quantitative Information Flow} (QIF) \cite{9da52c96268c449d85dacfcd15f27685}. In the QIF framework, a system is modeled as an information-theoretic channel that takes in a secret input and produces an observable output. The amount of information leakage is defined as the difference between the \emph{vulnerability} of the secret before and after passing through the channel—that is, the difference between the prior and posterior vulnerabilities, which quantify the adversary’s ability to perform a successful inference attack \cite{10.1109/CSF.2012.26, 10.1109/CSF.2014.29}.

In other words, within the QIF formulation, the attacker aims to infer the value of the secret itself. In our terminology, this corresponds to guessing $y_1$ in Game~\ref{guess_game} or $x_1$ in Game~\ref{guess_game2}, whereas our formulation focuses on guessing the \emph{position} of the secret (i.e., $\sigma^{-1}(1)$). We compare our definition with the QIF-based formulation of Game~\ref{guess_game} in Section \ref{sec_definition_basic}.

There have been studies exploring QIF in the context of differential privacy \cite{doi:10.3233/JCS-150528,8823711}; most notably, Reference \cite{10221886} was the first to examine QIF in the shuffle DP setting.
Interestingly, they compare the QIF of LDP, shuffle-only, and combined LDP+shuffle settings, and observe that, under an uninformed adversary (who does not know any individual’s data prior to accessing a dataset and assumes a uniform prior over datasets), the shuffling process contributes the majority of the privacy in $k$-RR.
However, QIF analysis in general is known to be technically challenging, and closed-form expressions are available only for a few special cases\cite{10221886}.

Specifically, their work focuses on an \emph{uninformed adversary}. In contrast, prior research typically considers a \emph{strong adversary}, who knows all individuals’ data except that of the target user. They provide an analytical treatment of the quantitative information flow for $k$-RR against an uninformed adversary, whereas the QIF of other mechanisms under an uninformed adversary remains unclear. In the case of a strong adversary, even the simple $k$-RR mechanism is analytically intractable; they only provide a brief discussion based on numerical experiments and leave a thorough analysis for future work.

In contrast, our work demonstrates that for re-identification attacks, the Bayesian advantage in shuffle DP admits a much simpler and unified analysis.

\section{Analysis in the Basic Setting}

\subsection{Bayesian optimal attack}\label{sec:bayesian_attack}

We now derive the success probability of an adversary $\mathcal{A}$ in the basic setting, where the distributions $P$ and $Q$ are known, and the adversary observes a randomly shuffled vector $\boldsymbol{y}_{\sigma} = \boldsymbol{y}'$ and outputs a guess $g$ for the index of the sample drawn from $P$.
\begin{align}
\Pr[&\sigma^{-1}(1) = g \mid \boldsymbol{y}_{\sigma} = \boldsymbol{y}'] \nonumber \\
&= \Pr[\sigma^{-1}(1) = g \mid \boldsymbol{y} = \boldsymbol{y}'_{\sigma^{-1}}] \label{eq1} \\
&= \frac{\Pr[\sigma^{-1}(1) = g \wedge \boldsymbol{y} = \boldsymbol{y}'_{\sigma^{-1}}]}{\Pr[\boldsymbol{y} = \boldsymbol{y}'_{\sigma^{-1}}]} \label{eq2} \\
&= \frac{\Pr[\sigma^{-1}(1) = g] \cdot \Pr[\boldsymbol{y} = \boldsymbol{y}'_{\sigma^{-1}} \mid \sigma^{-1}(1) = g]}{\sum_{\sigma_1 \in \mathcal{P}_n} \Pr[\sigma = \sigma_1] \cdot \Pr[\boldsymbol{y} = \boldsymbol{y}'_{\sigma^{-1}} \mid \sigma = \sigma_1]} \label{eq3} \\
&= \frac{\frac{1}{n} \cdot \frac{P(y'_g)}{Q(y'_g)} \cdot \prod_{j=1}^n Q(y'_j)}{\sum_{i=1}^n \frac{1}{n} \cdot \frac{P(y'_i)}{Q(y'_i)} \cdot \prod_{j=1}^n Q(y'_j)} \label{eq4} \\
&= \frac{\frac{P(y'_g)}{Q(y'_g)}}{\sum_{i=1}^n \frac{P(y'_i)}{Q(y'_i)}}. \label{eq5}
\end{align}

Here, $\sigma$ is a uniformly random permutation over $[n] = \{1, 2, \dots, n\}$. Equation~\eqref{eq1} holds because for each $i$, $y_i = y'_{\sigma^{-1}(i)}$. Equation~\eqref{eq2} follows from the definition of conditional probability, and Equation~\eqref{eq3} expands the denominator using the law of total probability.

Equation~\eqref{eq4} leverages two facts:  
(i) $\Pr[\sigma^{-1}(1) = g] = 1/n$ due to uniform randomness, and  
(ii) conditioned on $\sigma^{-1}(1) = g$, the value $y'_g$ is drawn from $P$, while all other $y'_j$ are drawn from $Q$. This gives:
\begin{align*}
\Pr[\boldsymbol{y} = \boldsymbol{y}'_{\sigma^{-1}} \mid \sigma^{-1}(1) = g]
&= P(y'_g) \cdot \prod_{j \ne g} Q(y'_j)\\
&= \frac{P(y'_g)}{Q(y'_g)} \cdot \prod_{j=1}^n Q(y'_j).
\end{align*}

In detail, denote $\mathcal{P}_n' = \{\sigma \mid \sigma^{-1}(1) = g\}$. Conditioning on $\sigma^{-1}(1) = g$ implies that $\sigma$ is uniformly sampled from $\mathcal{P}_n'$, where $|\mathcal{P}_n'| = (n-1)!$. For any $\sigma_* \in \mathcal{P}_n'$, we have
\[
\Pr[\boldsymbol{y} = \boldsymbol{y}'_{\sigma^{-1}} \mid \sigma = \sigma_*] = \Pr[\boldsymbol{y} = \boldsymbol{y}'_{\sigma^{-1}_*}] = P(y'_g) \cdot \prod_{j \ne g} Q(y'_j),
\]
which is independent of the choice of $\sigma_*$. Consequently,
\begin{align*}
\Pr[\boldsymbol{y} = \boldsymbol{y}'_{\sigma^{-1}} \mid \sigma^{-1}(1) = g] 
&= \sum_{\sigma_* \in \mathcal{P}_n'} \frac{1}{(n-1)!} \Pr[\boldsymbol{y} = \boldsymbol{y}'_{\sigma^{-1}_*}] \\
&= P(y'_g) \cdot \prod_{j \ne g} Q(y'_j).
\end{align*}

The same logic applies to each term in the denominator.

Equation~\eqref{eq5} reveals the form of the \textbf{Bayes-optimal strategy}: the adversary computes the likelihood ratio $P(y'_i) / Q(y'_i)$ for all $i \in [n]$, and ranks them in descending order. The index $g$ with the largest ratio is the most likely to be the sample from $P$.

Importantly, if for some $i$, $Q(y'_i) = 0$ but $P(y'_i) > 0$, we interpret $P(y'_i)/Q(y'_i) = \infty$, indicating that $y'_i$ must have come from $P$ with certainty. In that case, the attacker can identify the position $\sigma^{-1}(1) = i$ with probability 1.

\begin{corollary}\label{corollary1}
If $P = Q$, then for any $n \ge 1$, the success probability of any adversary $\mathcal{A}$ in the basic setting is exactly $1/n$. That is,
\[
\beta_n(P, P) = \frac{1}{n}, \qquad \mathsf{Adv}_n^{+}(P, P) = 0.
\]
\end{corollary}

\begin{proof}
From Equation~\eqref{eq5}, when $P = Q$, the likelihood ratio $P(y'_i)/Q(y'_i) = 1$ for all $i \in [n]$. Therefore, the adversary sees $n$ identical scores, and no position is statistically distinguishable from the others. Hence, any strategy reduces to uniform guessing, and the success probability is $1/n$.
\end{proof}

\subsection{Theoretical derivation of $\beta_n(P,Q)$}

As we will demonstrate, the core of the theoretical computation of $\beta_n(P, Q)$ lies in an insightful shift in perspective. The crucial idea is to transition from working directly with the two distributions $P$ and $Q$ to studying the distributions of their likelihood ratio $\frac{P(y)}{Q(y)}$.

We define two cumulative distribution functions:
\begin{subequations}
    \begin{align}
        F(t) &= \Pr_{y \sim P}\left[\frac{P(y)}{Q(y)} \le t\right], \label{def_f} \\
        G(t) &= \Pr_{y \sim Q}\left[\frac{P(y)}{Q(y)} \le t\right], \label{def_g}
    \end{align}
\end{subequations}
where $\frac{P(y)}{Q(y)}$ is the likelihood ratio of $y$, a quantity with broad applications in information theory~\cite{Polyanskiy22}. Let $f(t)$ and $g(t)$ denote the corresponding probability density functions of $F$ and $G$, respectively.

We also introduce two useful notations\footnote{According to the definition, $M$ may be $+\infty$ in some cases.}:
\[
f(+\infty) := \Pr_{y \sim P}[Q(y) = 0], \qquad M := \sup_{y : Q(y) \ne 0} \frac{P(y)}{Q(y)}.
\]

Recall that the Bayes-optimal attacker sorts the entries by descending likelihood ratio. The distribution of the likelihood ratio of a sample drawn from $P$ is exactly $f$. For example, if the likelihood ratio of $y_1$ is $1.4$, then the attacker will succeed in the first guess if and only if all other $y_i$, $i=2,3,\dots,n$, have likelihood ratios less than $1.4$. The probability of this event is $f(1.4) \cdot G^{n-1}(1.4)$. Integrating over all possible values of the likelihood ratio yields:
\begin{align}
\beta_n(P, Q) = f(+\infty) + \int_0^M f(t) G^{n-1}(t) \, \mathrm{d}t. \label{epsilon11}
\end{align}

A similar approach yields the expression for $\beta_n^k(P, Q)$. By the definition of $\beta_n^k(P, Q)$ and the optimal strategy, the Bayesian adversary wins in $\mathbf{GuessGame}^{P,Q,n}_{\mathcal{A}}(k)$ if and only if the number of $y_i$ ($i=2,\dots,n$) with likelihood ratios greater than that of $y_1$ is fewer than $k$.

\begin{theorem}
The Bayesian success probability with $k$ guesses is given by:
\begin{align}
\beta&_n^k(P,Q) = \nonumber\\
&f(+\infty) + \sum_{j=1}^{k} \int_0^M \binom{n-1}{j-1} f(t) G^{n-j}(t) \left(1 - G(t)\right)^{j-1} \, \mathrm{d}t.
\label{f6}
\end{align}
\end{theorem}

\begin{proof}
The adversary succeeds in the $j$-th guess if and only if exactly $j-1$ of the $n-1$ decoy samples have likelihood ratios greater than that of $y_1$. The integrand accounts for all configurations where this occurs, and the integral aggregates over all possible threshold values $t$. Summing over $j$ from $1$ to $k$ gives the desired result.
\end{proof}

The following property will be essential in simplifying these expressions:

\begin{property}\label{prop0}
The density functions $f$ and $g$ defined above satisfy:
\begin{align}
f(t) = t \cdot g(t). \label{f2}
\end{align}
\end{property}

\begin{proof}
Let $S_t = \{y \in \mathbb{Y} \mid \frac{P(y)}{Q(y)} = t\}$. Then:
\begin{align*}
\frac{f(t)}{g(t)} 
= \frac{P(S_t)}{Q(S_t)} 
= \frac{\sum_{y \in S_t} P(y)}{\sum_{y \in S_t} Q(y)} 
= \frac{\sum_{y \in S_t} t \cdot Q(y)}{\sum_{y \in S_t} Q(y)} = t. 
\end{align*}
\end{proof}

\begin{theorem}\label{theorem2}
Let $P$ and $Q$ be probability distributions such that 
\[
M := \sup_{y : Q(y) \ne 0} \frac{P(y)}{Q(y)} < \infty.
\]
Then the Bayesian success probability $\beta_n(P, Q)$ can be expressed as:
\begin{align}
\beta_n(P, Q) = \frac{1}{n} \left( M - \int_0^M G^n(t) \, \mathrm{d}t \right) + f(+\infty). \label{f3}
\end{align}
\end{theorem}

\begin{proof}
Using integration by parts and Property~\ref{prop0}:
\begin{align*}
\beta_n(P,Q) 
&= \int_0^M f(t) G^{n-1}(t) \, \mathrm{d}t + f(+\infty) \\
&= \int_0^M t \cdot g(t) \cdot G^{n-1}(t) \, \mathrm{d}t + f(+\infty) \\
&= \int_0^M t \cdot G^{n-1}(t) \, \mathrm{d}G(t) + f(+\infty) \\
&= \frac{1}{n} \int_0^M t \, \mathrm{d}G^n(t) + f(+\infty) \\
&= \frac{1}{n} \left[ t G^n(t) \big|_0^M - \int_0^M G^n(t) \, \mathrm{d}t \right] + f(+\infty) \\
&= \frac{1}{n} \left( M - \int_0^M G^n(t) \, \mathrm{d}t \right) + f(+\infty). \qedhere
\end{align*}
\end{proof}

Equation~\eqref{f3} clearly shows how $\beta_n(P, Q)$ scales with $n$: 
\begin{theorem}\label{theorem:asymptotic}
Let $P$ and $Q$ be probability distributions such that $M := \sup\limits_{y : Q(y) \ne 0} \frac{P(y)}{Q(y)} < \infty$.
Then the Bayesian success probability
\[
\beta_n(P, Q) = f(+\infty) + \frac{M}{n} - \ell(n),
\]
where $\ell(n) \ge 0 \text{ and } \ell(n) = o\left(\frac{1}{n}\right)$.
\end{theorem}
\begin{proof}
From Theorem~\ref{theorem2}, we have
\[
\beta_n(P, Q) = f(+\infty) + \frac{M}{n} - \frac{1}{n} \int_0^M G^n(t) \, \mathrm{d}t.
\]
Define $\ell(n) := \frac{1}{n} \int_0^M G^n(t) \, \mathrm{d}t$. Since $0 \le G(t) \le 1$ for all $t$ (as $G$ is a cumulative distribution function), it follows that $0 \le \ell(n) \le \frac{M}{n}$. Moreover, as $G^n(t) \to 0$ pointwise for $t < M$ when $n \to \infty$, and the integrand is dominated by an integrable function, the dominated convergence theorem implies:
\[
\ell(n) = \frac{1}{n} \int_0^M G^n(t) \, \mathrm{d}t = o\left(\frac{1}{n}\right). \qedhere
\]
\end{proof}

There is a subtle point in the derivation of $\beta_n(P, Q)$ that warrants clarification: when multiple entries $y_i$ share the same maximum likelihood ratio, how should ties be resolved? This directly determines whether the cumulative function $G(t)$ in the expression $f(t) G^{n-1}(t)$ includes the probability mass at $t$.

We formalize this subtlety in the following result:

\begin{theorem}\label{therm3}
Let $P$ and $Q$ be discrete probability distributions. Then the Bayesian success probability is given by:
\[
\beta_n(P, Q) = f(+\infty) + \frac{1}{n} \sum_{t} t \left[ G^n(t) - G^n(t^-) \right],
\]
where $G(t^-)$ denotes the left-hand limit of $G$ at $t$, i.e., $G(t^-) = \Pr_{y \sim Q}\left[\frac{P(y)}{Q(y)} < t\right]$.
\end{theorem}

The proof is provided in Appendix \ref{appendix_1}, which proceeds via a combinatorial calculation.
This result is the discrete analogue of the continuous case:
\[
\beta_n(P, Q) = f(+\infty) + \frac{1}{n} \int_0^M t \, \mathrm{d}G^n(t).
\]


\subsection{Case studies}
In this section, we demonstrated how to compute $\beta_n(P, Q)$ under two representative examples. In both cases, the decoy distribution $Q$ is chosen to be uniform—this choice is made not only to facilitate theoretical analysis, but also because the uniform distribution is commonly adopted in practice when little is known about the target distribution $P$.

\begin{exmp}\label{exmp1}
Consider the following two continuous probability distributions defined on the interval $[0, 1]$:
\begin{align*}
    P(x) &= x + 0.5, \quad &0 \le x \le 1, \\
    Q(x) &= 1, \quad &0 \le x \le 1.
\end{align*}
In this case, the likelihood ratio $\frac{P(x)}{Q(x)} = x + 0.5$, and it follows:
\begin{align*}
    f(t) = t, \quad &0.5 \le t \le 1.5,\\
    g(t) = 1, \quad &0.5 \le t \le 1.5.
\end{align*}

Using the formula in Theorem~\ref{theorem2}, we compute:
\begin{align*}
\beta_n(P, Q) 
&= \frac{1}{n}\left(M - \int_0^M G^n(t) \, \mathrm{d}t\right) 
= \frac{1}{n} \left(1.5 - \frac{1}{n+1} \right) \\
&= \frac{1.5}{n} - \frac{1}{n(n+1)}.
\end{align*}
Hence, $\beta_n(P, Q) = \Theta(1.5/n)$.

We can also compute the generalized success probability for $k$ guesses:
\[
\beta_n^k(P, Q) = \frac{3n + 2}{2n(n+1)} \cdot k - \frac{k^2}{2n(n+1)}.
\]
This is a quadratic function of $k$. The function $\beta_n^k(P, Q)$ versus $k$ for $n = 20$ is illustrated in Figure~\ref{fig3}.
\end{exmp}

\begin{figure}[tbp]
    \centering
    \includegraphics[scale=0.45]{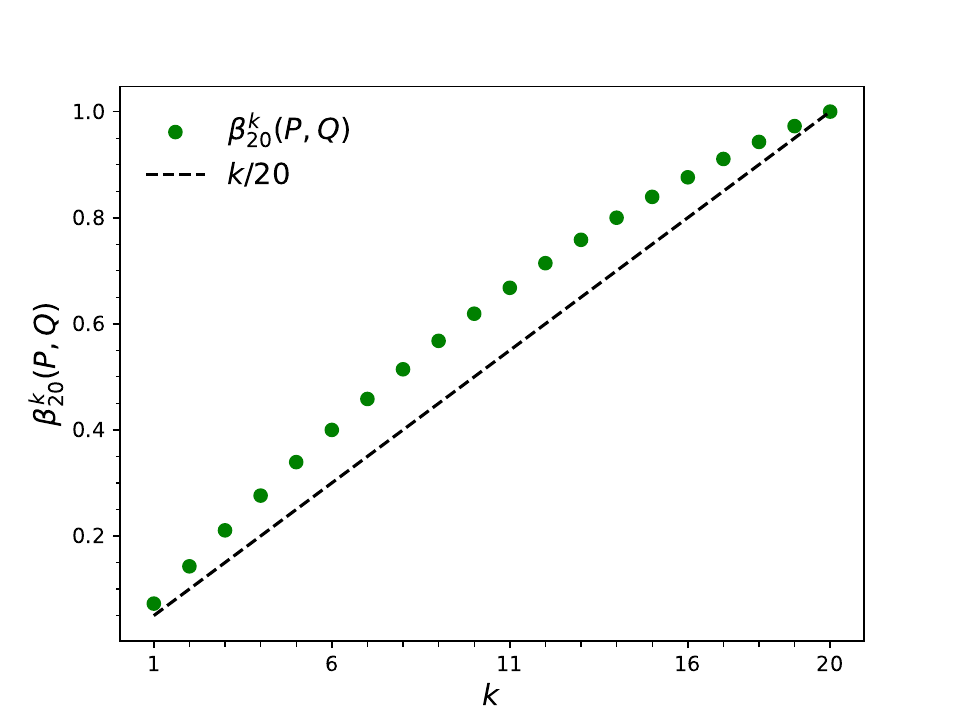}
    \caption{$\beta_n^k(P,Q)$ as a function of $k$ for $n=20$ in Example~\ref{exmp1}}
    \label{fig3}
\end{figure}

The next example involves the Honeyword setting. To model real-world password behavior, we assume a non-uniform distribution $P$ for real passwords. Empirical studies have shown that human-chosen passwords approximately follow a Zipf distribution~\cite{11,12}, where the $r$-th most probable password has probability:
\[
P(pw_r) = \frac{r^{-\alpha}}{\sum_{i=1}^m i^{-\alpha}},
\]
with parameter $0 < \alpha < 1$ and $m = |\mathcal{PW}|$ denoting the size of the password space. This reflects the intuition that a small number of high-ranked passwords dominate the distribution.

\begin{exmp}\label{exmp2}
Let $P$ be a Zipf distribution with parameter $\alpha$ over $[m]$, and let $Q$ be the uniform distribution over $[m]$. Define $S = \sum_{j=1}^m j^{-\alpha} \approx \frac{1}{1 - \alpha} m^{1 - \alpha}$. Then, the Bayesian success probability $\beta_n(P, Q)$ is given by:
\begin{align*}
\beta_n(P, Q) 
&= \frac{1}{n} \sum_t t \left[G^n(t) - G^n(t^-)\right] \\
&= \frac{1}{n} \sum_{i=1}^m m \cdot \frac{(m - i + 1)^{-\alpha}}{S} \left[\left(\frac{i}{m}\right)^n - \left(\frac{i-1}{m}\right)^n\right],
\end{align*}
where we use the fact that under $P$ and $Q$, the likelihood ratio $\frac{P(i)}{Q(i)}$ equals $m \cdot \frac{(m - i + 1)^{-\alpha}}{S}$.

To simplify this expression, we approximate the finite difference by a derivative:
\[
\left(\frac{i}{m}\right)^n - \left(\frac{i-1}{m}\right)^n \approx \frac{n}{m} \left(\frac{i}{m}\right)^{n-1},
\]
and interpret the summation as a Riemann sum:
\begin{align*}
\beta_n(P, Q) 
&\approx \frac{m^{1 - \alpha}}{S} \sum_{i=1}^m \frac{1}{m} \left(1 - \frac{i+1}{m}\right)^{-\alpha} \left(\frac{i}{m}\right)^{n - 1}\\
&\approx \frac{m^{1 - \alpha}}{S} \int_0^1 (1 - t)^{-\alpha} t^{n - 1} \, \mathrm{d}t.
\end{align*}

This integral corresponds to the Beta function:
\[
\beta_n(P, Q) \approx \frac{m^{1 - \alpha}}{S} \cdot B(1 - \alpha, n),
\]
where
\[
B(a, b) = \int_0^1 t^{a - 1} (1 - t)^{b - 1} \, \mathrm{d}t
\]
is the Beta function. Substituting the approximation of $S$, we obtain:
\[
\beta_n(P, Q) \approx (1 - \alpha) \cdot B(1 - \alpha, n).
\]
\end{exmp}

Likewise, we can compute the multi-guess success probability:
\[
\beta_n^k(P, Q) = \sum_{j=1}^k (1 - \alpha) \cdot \binom{n-1}{j-1} \cdot B(j - \alpha, n + 1 - j).
\]

Assuming $\alpha = 0.7$~\cite{21}, the plot of $\beta_n(P, Q)$ versus $n$ is shown in Figure~\ref{fig1}.
The graph of $\beta_n^k(P, Q)$ as a function of $k$ for $n = 20$ is shown in Figure~\ref{fig111}.
Numerical evaluations show that when $P$ is a Zipf distribution with parameter $0.7$ and $Q$ is uniform, the Bayes-optimal adversary achieves a success probability $\beta_n(P, Q) < 0.2$ only if $n \ge 150$. Moreover, for $n = 20$, the success probability of the optimal adversary within the first $k = 3$ guesses, i.e., $\beta_n^3(P, Q)$, already exceeds $0.5$.

Our theoretical results offer quantitative guidance for estimating the attack success probability of a Bayes-optimal adversary in this setting and for selecting the number of decoys necessary to ensure adequate security.

\begin{figure}[t]
    \centering
    \includegraphics[scale=0.45]{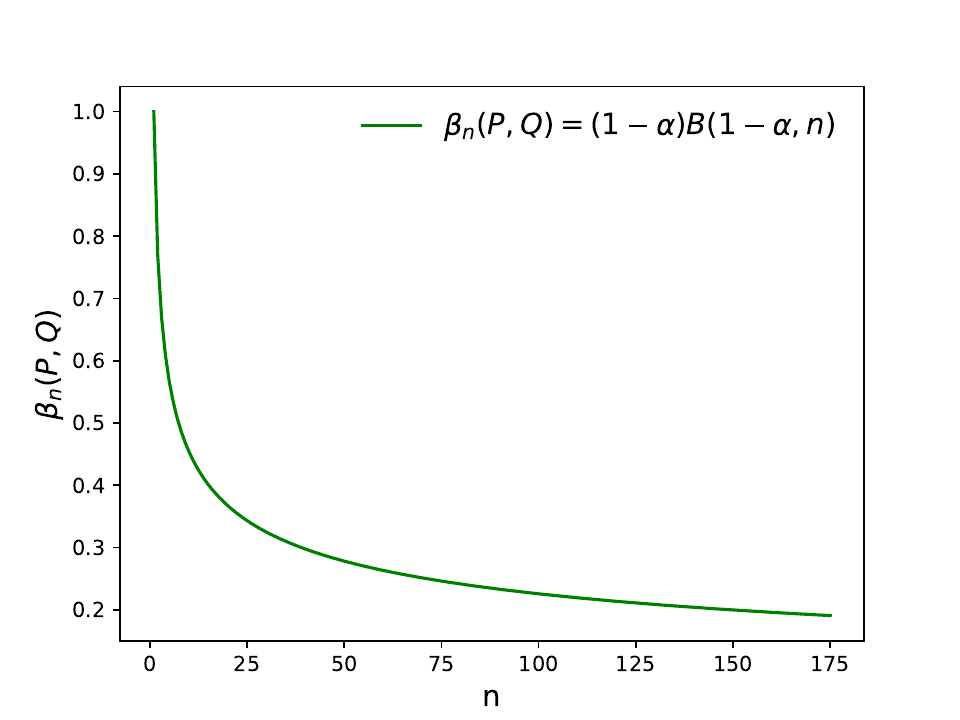}
    \caption{$\beta_n(P,Q)$ vs. $n$ where $P=Zipf(0.7)$ and $Q$ is uniform}
    \label{fig1}
\end{figure}

\begin{figure}[t]
    \centering
    \includegraphics[scale=0.45]{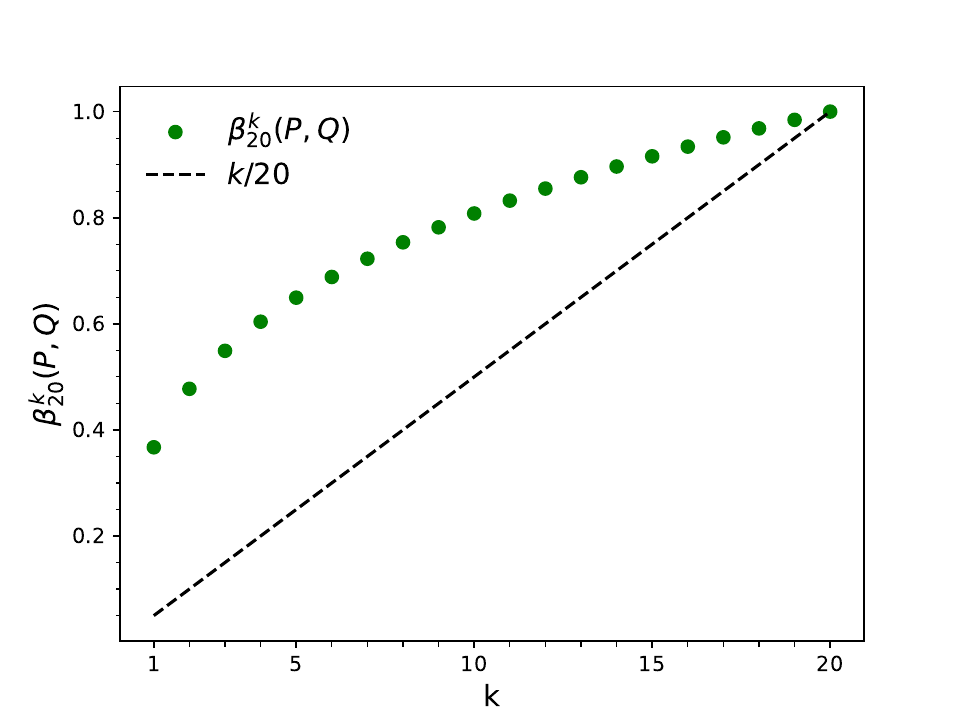}
    \caption{$\beta_{20}^k(P,Q)$ vs. $k$ where $P=Zipf(0.7)$ and $Q$ is uniform}
    \label{fig111}
\end{figure}

\section{Connections Between additive Bayesian advantage and total variation}\label{sec:connections}

\subsection{Main results}
Many cryptographic definitions implicitly rely on distinguishing games between two probability distributions. A classical setting is as follows: the challenger samples $m_0 \sim P$ and $m_1 \sim Q$, selects a random bit $b \in \{0,1\}$, and sends $m_b$ to the adversary, who must then guess the value of $b$ by outputting $b'$. The adversary’s success probability is $\Pr[b' = b]$, and it is well known that the maximal advantage over random guessing is characterized by:
\[
\max_{\mathcal{A}} \Pr[b' = b] = \frac{1}{2} + \frac{1}{2} \Delta(P, Q),
\]
where $\Delta(P, Q) := \frac{1}{2} \sum_{x \in \mathbb{X}} |P(x) - Q(x)|$ denotes the total variation distance between $P$ and $Q$.

Now consider a stronger challenge where both $m_0$ and $m_1$ are given to the adversary, who must decide which one came from $P$ and which from $Q$. This corresponds to our re-identification setting $\mathbf{GuessGame}^{P,Q,n}_{\mathcal{A}}(k)$ with $n = 2$ and $k = 1$. In this case, the optimal success probability is given by $\beta_2(P, Q)$, and the Bayesian advantage becomes $\mathsf{Adv}_2^+(P, Q) = \beta_2(P, Q) - \frac{1}{2}$.

The following theorem establishes a tight relationship between $\mathsf{Adv}_n^+(P, Q)$ and the total variation distance $\Delta(P, Q)$.

\begin{theorem}\label{crypto-game}
For any distributions $P, Q$ and integer $n \ge 1$, the Bayesian advantage satisfies:
\[
\frac{1}{n} \cdot \Delta(P, Q) \le \mathsf{Adv}_n^+(P, Q) \le \Delta(P, Q).
\]
\end{theorem}

The proof is deferred to Appendix~\ref{appendix_2}. The lower bound is established via an explicit construction of an adversary's strategy, whereas the upper bound is more involved and requires several analytical techniques.

\subsection{Tightness analysis}
\textbf{Tightness of the upper bound.}
Consider the distributions $P(0) = p$, $P(1) = 1 - p$ and $Q(0) = 0$, $Q(1) = 1$. In this case, the total variation distance is $\Delta(P, Q) = p$, and a direct computation using Theorem~\ref{therm3} yields:
\[
\beta_n(P, Q) = p + \frac{1 - p}{n},\quad \mathsf{Adv}_n^+(P, Q) = \left(1 - \frac{1}{n}\right)p.
\]
This demonstrates that for any fixed total variation distance $\Delta(P, Q) = p$, there exist distributions $P$ and $Q$ such that:
\[
\mathsf{Adv}_n^+(P, Q) = \left(1 - \frac{1}{n} \right) \Delta(P, Q).
\]
Therefore, the upper bound $\mathsf{Adv}_n^+(P, Q) \le \Delta(P, Q)$ is tight up to a factor of $1/n$, and cannot be improved in general.

\textbf{Tightness of the lower bound.}
Consider the distributions $P(0) = 0$, $P(1) = 1$, and $Q(0) = p$, $Q(1) = 1 - p$. Then the total variation distance is $\Delta(P, Q) = p$. A direct computation shows:
\[
\mathsf{Adv}_n^+(P, Q) = \frac{1}{(1 - p)n}(1 - p^n) - \frac{1}{n} = \frac{p}{(1 - p)n} - \frac{p^n}{(1 - p)n}.
\]
For any fixed $n \ge 2$, as $p \to 0$, we have $\mathsf{Adv}_n^+(P, Q) = \Theta\left( \frac{p}{n} \right)$, which matches the lower bound $\frac{1}{n} \Delta(P, Q)$ up to a vanishing multiplicative factor.

\section{Analysis of Bayesian Advantage in Shuffle DP}

\subsection{Reduction to the basic setting}

As discussed earlier, unlike the basic setting—where the Bayesian attack admits a closed-form analytical expression—the primary challenge in analyzing the shuffle DP model stems from the fact that the output distributions $\mathcal{R}(x_i)$ for users $i = 2, 3, \dots, n$ are generally non-identical.

Fortunately, this difficulty can be circumvented by reducing the shuffle DP analysis to the basic setting through a probabilistic decomposition of the local randomizer $\mathcal{R}$. Specifically, suppose $\mathcal{R}$ satisfies the following mixture structure:
\[
\forall x \in \mathbb{X}:\quad \mathcal{R}(x) = \gamma Q^{\text{com}} + (1 - \gamma) \cdot \text{LO}(x),
\]
where $Q^{\text{com}}$ is a distribution shared across all users (the \emph{common component}), and $\text{LO}(x)$ is an input-dependent \emph{left-over} distribution.

Under this decomposition, the re-identification problem reduces to the basic setting with $P = \mathcal{R}(x_1)$, $Q = Q^{\text{com}}$, and a random number of decoys $N \sim 1 + \mathrm{Bin}(n - 1, \gamma)$, where $\mathrm{Bin}(n - 1, \gamma)$ denotes the binomial distribution with $n - 1$ trials and success probability $\gamma$. This reformulation is formalized in Game~\ref{guess_game3}.

We define:
\begin{align*}
&\psi(\mathcal{R}, n, \gamma, Q^{\text{com}}, x_1) :=\\
&\quad \quad\quad\max_{\mathcal{A}} \Pr\left[\mathbf{GuessGame}^{\mathcal{R}, n, \gamma, Q^{\text{com}}}_{\mathcal{A}}(x_1) = {\mathbf{True}}\right].
\end{align*}

\begin{theorem}\label{main_theorem}
If the local randomizer $\mathcal{R}$ admits a decomposition of the form
\[
\forall x \in \mathbb{X}:\quad \mathcal{R}(x) = \gamma Q^{\text{com}} + (1 - \gamma) \cdot \mathrm{LO}(x),
\]
then
\[
\beta_n(\mathcal{R}, x_1) \le \psi(\mathcal{R}, n, \gamma, Q^{\text{com}}, x_1).
\]
\end{theorem}

\begin{proof}
Let $\boldsymbol{x}_{-1} = (x_2, x_3, \dots, x_n)$ denote the inputs of users $2$ to $n$, and let $\boldsymbol{y} = (y_1, \dots, y_n)$ with $y_i \sim \mathcal{R}(x_i)$ be their corresponding outputs. Define a binary vector $\boldsymbol{b} = (b_2, \dots, b_n) \in \{0,1\}^{n-1}$ such that $b_i = 1$ indicates $y_i \sim Q^{\text{com}}$, and $b_i = 0$ indicates $y_i \sim \text{LO}(x_i)$. Let $I = \{ i \in [2,n] \mid b_i = 0 \}$ be the index set of outputs not sampled from $Q^{\text{com}}$.

Game~\ref{guess_game4} corresponds to a variant of Game~\ref{guess_game2} in which the adversary is additionally provided the positions $\sigma^{-1}(I)$ of samples drawn from the non-common distributions $\text{LO}(x_i)$. We proceed by a hybrid argument. If the adversary is given $\sigma^{-1}(I)$, then the optimal strategy reduces to that in Game~\ref{guess_game3}. If not, the adversary operates as in Game~\ref{guess_game2}.

Let $\mathcal{A}$ be any adversary in the original game $\mathbf{GuessGame}^{\mathcal{R}, n}_{\mathcal{A}}(\boldsymbol{x})$. We define a new adversary $\mathcal{B}$ for Game~\ref{guess_game4} that ignores the extra information $\sigma^{-1}(I)$ and simply invokes $\mathcal{A}$ on the permuted outputs $\boldsymbol{y}_\sigma$. Since $\mathcal{A}$ does not utilize the additional information, the success probability remains unchanged. Hence,
\begin{align*}
\forall \mathcal{A},\exists \mathcal{B}, \forall \boldsymbol{x}&\in \{\boldsymbol{x}\in \mathbb{X}^n \mid\boldsymbol{x}[1]=x_1\},\\
\Pr&\left[\mathbf{GuessGame}^{\mathcal{R}, n}_{\mathcal{A}}(\boldsymbol{x}) = {\mathbf{True}}\right]\\
&=\Pr\left[\mathbf{GuessGame2}^{\mathcal{R}, n, \gamma, Q^{\text{com}}}_{\mathcal{B}}(x_1) = {\mathbf{True}}\right].
\end{align*}
Taking the maximum over all adversaries $\mathcal{A}$ and corresponding $\mathcal{B}$, we obtain:
\[
\beta_n(\mathcal{R}, x_1) \le \psi(\mathcal{R}, n, \gamma, Q^{\text{com}}, x_1). \qedhere
\]
\end{proof}

\renewcommand{\algorithmcfname}{Game}
\begin{algorithm}[t]
$y_1 \gets_{\mathcal{R}(x_1)} \mathbb{Y}$\\
$N \sim 1+\mathrm{Bin}(n-1,\gamma)$\\
$y_2, y_3, \dots, y_N \gets_{Q^{\text{com}}} \mathbb{Y}$\\
$\boldsymbol{y} \gets (y_1, y_2, \dots, y_N)$\\
$\sigma \gets_{\$} \mathcal{P}_N$ \tcp{Sample a random permutation}
$\boldsymbol{y}_\sigma \gets (y_{\sigma(1)}, y_{\sigma(2)}, \dots, y_{\sigma(N)})$\\
$g \gets \mathcal{A}(\boldsymbol{y}_\sigma)$\\
\textbf{return} $\sigma^{-1}(1) = g$
\caption{$\mathbf{GuessGame}^{\mathcal{R},n,\gamma,Q^{\text{com}}}_{\mathcal{A}}(x_1)$}
\label{guess_game3}
\end{algorithm}

\renewcommand{\algorithmcfname}{Game}
\begin{algorithm}[t]
$y_1 \gets_{\mathcal{R}(x_1)} \mathbb{Y}$\\
$I \gets \emptyset$\\
\For{$i = 2$ \textbf{to} $n$}{
    $b_i \sim \mathrm{Bern}(\gamma)$\\
    \If{$b_i = 1$}{
        $y_i \gets_{Q^{\text{com}}} \mathbb{Y}$
    }
    \Else{
        $y_i \gets_{\text{LO}(x_i)} \mathbb{Y}$; \quad $I \gets I \cup \{i\}$
    }
}
$\boldsymbol{y} \gets (y_1, y_2, \dots, y_n)$\\
$\sigma \gets_{\$} \mathcal{P}_n$ \tcp{Sample a random permutation}
$\boldsymbol{y}_\sigma \gets (y_{\sigma(1)}, y_{\sigma(2)}, \dots, y_{\sigma(n)})$\\
$g \gets \mathcal{A}(\boldsymbol{y}_\sigma, \sigma^{-1}(I))$\\
\textbf{return} $\sigma^{-1}(1)=g$
\caption{$\mathbf{GuessGame2}^{\mathcal{R},n,\gamma,Q^{\text{com}}}_{\mathcal{A}}(x_1)$}
\label{guess_game4}
\end{algorithm}

\begin{theorem}\label{theorem_compute_psi}
The Bayesian success probability in the reduced game satisfies:
\begin{align*}
&\psi(\mathcal{R}, n, \gamma, Q^{\text{com}}, x_1) \\
&= \sum_{m=0}^{n-1} \binom{n-1}{m} \gamma^m (1 - \gamma)^{n-1 - m} \cdot \beta_{m+1}(\mathcal{R}(x_1), Q^{\text{com}}).
\end{align*}
\end{theorem}

\begin{proof}
This follows directly from the definition of $\psi(\cdot)$ in Game~\ref{guess_game3} and the definition of $\beta_n(P, Q)$.

In Game~\ref{guess_game3}, the total number of samples drawn from $Q^{\text{com}}$ (excluding the one from $P = \mathcal{R}(x_1)$) is a random variable $m \sim \mathrm{Bin}(n-1, \gamma)$. For each such $m$, the attacker's success probability is $\beta_{m+1}(\mathcal{R}(x_1), Q^{\text{com}})$, since there are $m$ decoys and one target. Taking the expectation over the binomial distribution gives the result.
\end{proof}

\begin{figure}
    \centering
    \includegraphics[width=0.9\linewidth]{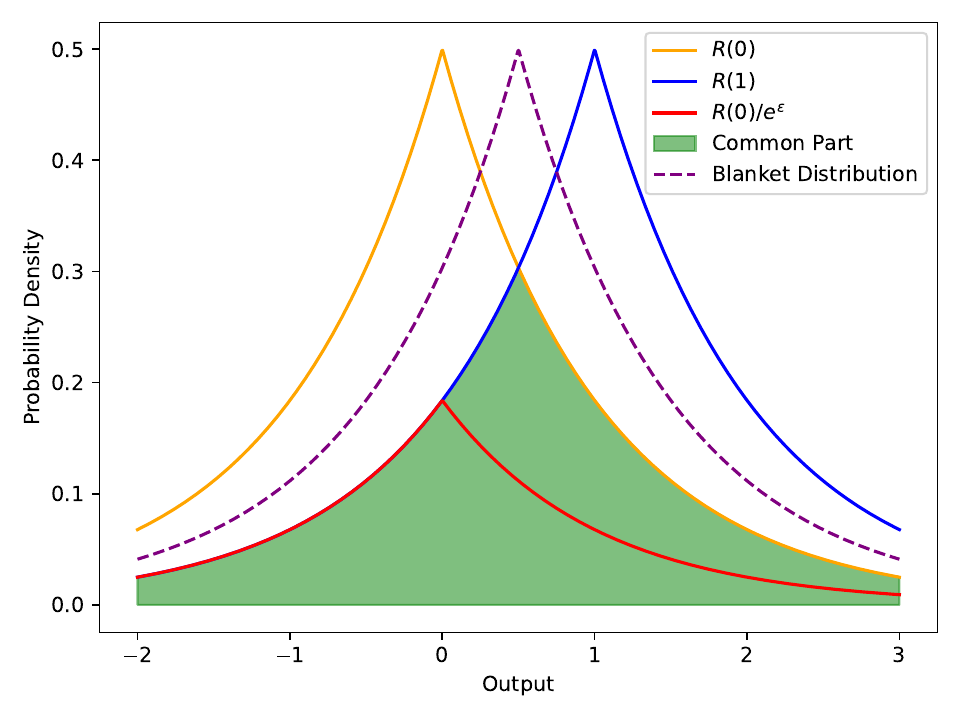}
    \caption{Decomposition methods for $1$-DP Laplace mechanism \cite{Dwork2006} on $\{0,1\}$}
    \label{fig:decomposition}
\end{figure}

\textbf{Possible decompositions.} The decomposition allowed by Theorem~\ref{main_theorem} admits multiple choices, among which the most commonly used are the \emph{clone} and \emph{blanket} decompositions, which we introduce separately below. For every $\varepsilon$-DP local randomizer, both of these decompositions are guaranteed to exist.

For an $(\varepsilon,\delta)$-DP local randomizer $\mathcal{R}$, the clone decomposition may fail to exist (i.e., $\gamma = 0$ when $Q^{\text{com}} = \mathcal{R}(x_1)$). The blanket decomposition, on the other hand, always exists but generally lacks a well-structured form, which introduces additional challenges. We leave the analysis of this scenario as future work.

\subsection{Analysis via the clone decomposition}
The \emph{clone} technique, originally proposed in~\cite{Feldman2021} for analyzing privacy amplification in shuffle DP, can also be applied to the analysis of re-identification attacks in the shuffle model. 
Specifically, for any $\varepsilon$-DP local randomizer $\mathcal{R}$, the clone method ensures the following decomposition:
\[
\forall x \in \mathbb{X}, \quad \mathcal{R}(x) = e^{-\varepsilon} \cdot \mathcal{R}(x_1) + (1 - e^{-\varepsilon}) \cdot \mathrm{LO}(x).
\]
This decomposition follows directly from the $\varepsilon$-DP property, which implies that for all $x \in \mathbb{X}$ and $y \in \mathbb{Y}$,
\[
\Pr[\mathcal{R}(x)=y] \ge e^{-\varepsilon} \cdot\Pr[\mathcal{R}(x_1)=y],
\]
thereby ensuring the existence of $\mathrm{LO}(x_i)$ as a valid residual distribution.

\begin{theorem}\label{theorem_clone}
Let $\mathcal{R}$ be an $\varepsilon$-differentially private local randomizer. Then, in the single-message shuffle model, the Bayesian success probability of any re-identification attack (as defined in Game~\ref{guess_game2}) satisfies:
\[
\beta_n(\mathcal{R}) \le \frac{e^\varepsilon}{n}.
\]
\end{theorem}

\begin{proof}
By Theorem~\ref{main_theorem}, using the \emph{clone decomposition}, we have:
\[
\beta_n(\mathcal{R}, x_1) \le \psi(\mathcal{R}, n, \gamma, \mathcal{R}(x_1), x_1),
\]
where we define $\gamma = e^{-\varepsilon}$. From Theorem~\ref{theorem_compute_psi}, and noting that $\beta_{m+1}(\mathcal{R}(x_1), \mathcal{R}(x_1)) = \frac{1}{m+1}$ (by Corollary~\ref{corollary1}), it follows that:
\begin{align*}
\psi(\mathcal{R},& n, \gamma, \mathcal{R}(x_1), x_1) \\
&= \sum_{m=0}^{n-1} \binom{n-1}{m} \gamma^m (1 - \gamma)^{n - 1 - m} \cdot \frac{1}{m + 1} \\
&= \frac{1}{\gamma n} \sum_{m=0}^{n-1} \binom{n}{m + 1} \gamma^{m + 1} (1 - \gamma)^{n - 1 - m} \\
&= \frac{1}{\gamma n} \left[1 - (1 - \gamma)^n\right] \\
&\le \frac{1}{\gamma n} = \frac{e^\varepsilon}{n}.
\end{align*}
This concludes the proof.
\end{proof}

\subsection{Analysis via the blanket decomposition}\label{sec_blanket}

Another form of decomposition in the shuffle DP setting is the \emph{blanket decomposition}, which extracts the \emph{maximum} common component across all $\mathcal{R}(x_i)$ \cite{Balle2019}. Specifically, the blanket distribution $Q^{\mathrm{B}}$ is defined pointwise as:
\[
Q^{\mathrm{B}}(y) := \frac{1}{\alpha} \cdot \inf_{x \in \mathbb{X}} \Pr[\mathcal{R}(x) = y], \
\alpha := \sum_{y \in \mathbb{Y}} \inf_{x \in \mathbb{X}} \Pr[\mathcal{R}(x) = y].
\]
This yields a valid convex decomposition:
\[
\mathcal{R}(x) = \alpha \cdot Q^{\mathrm{B}} + (1 - \alpha) \cdot \mathrm{LO}^{\mathrm{B}}(x),
\]
where $\mathrm{LO}^{\mathrm{B}}(x)$ is a residual distribution specific to each input $x$.

The blanket decomposition is known to be optimal in deriving privacy amplification bounds under shuffle DP~\cite{su2025decompositionbasedoptimalboundsprivacy}. We now demonstrate that it is also optimal for upper bounding the success probability of re-identification attacks among all decomposition-based methods.

\begin{theorem}\label{therm_optimal}
Let $Q^{\mathrm{B}}$ and $\alpha$ be the blanket distribution and its associated mass coefficient for local randomizer $\mathcal{R}$. If $\mathcal{R}$ also admits a decomposition:
\[
\forall x\in \mathbb{X}: \quad \mathcal{R}(x) = \gamma Q^{\text{com}} + (1 - \gamma) \mathrm{LO}(x_i),
\]
then
\[
\beta_n(\mathcal{R}, x_1) \le \psi(\mathcal{R}, n, \alpha, Q^{\mathrm{B}}, x_1) \le \psi(\mathcal{R}, n, \gamma, Q^{\text{com}}, x_1).
\]
\end{theorem}

\begin{proof}
The first inequality follows directly from Theorem~\ref{main_theorem} by choosing $\gamma = \alpha$ and $Q^{\text{com}} = Q^{\mathrm{B}}$.

To prove the second inequality, we use a hybrid argument. Since $\Pr[\mathcal{R}(x)=y] \ge \gamma Q^{\text{com}}(y)$ for all $x$ and $y$ by assumption, we have:
\[
\inf_{x} \Pr[\mathcal{R}(x) = y] \ge \gamma Q^{\text{com}}(y) \quad \Rightarrow \quad \alpha Q^{\mathrm{B}}(y) \ge \gamma Q^{\text{com}}(y).
\]
This implies a decomposition of $Q^{\mathrm{B}}$:
\[
Q^{\mathrm{B}} = \frac{\gamma}{\alpha} Q^{\text{com}} + \left(1 - \frac{\gamma}{\alpha} \right) Q^{\ell},
\]
for some distribution $Q^{\ell}$.

In Game~\ref{guess_game5}, users $i \ge 2$ sample from $Q^{\text{com}}$ with probability $\gamma$, and from $Q^{\ell}$ with probability $\alpha - \gamma$, which together form $Q^{\mathrm{B}}$ used with probability $\alpha$ overall. The game also reveals the shuffled positions $\sigma^{-1}(I)$ of users who sampled from $Q^{\ell}$, allowing the adversary to filter out those entries and behave identically to Game~\ref{guess_game3}.

If this additional information is hidden, the adversary is effectively operating in the game 
\[
\mathbf{GuessGame}^{\mathcal{R}, n, \alpha, Q^{\mathrm{B}}}_{\mathcal{A}}(x_1).
\]
Thus, any attacker for the blanket game can be simulated in Game~\ref{guess_game5} (by ignoring $\sigma^{-1}(I)$), while the converse is not true. Therefore:
\[
\psi(\mathcal{R}, n, \alpha, Q^{\mathrm{B}}, x_1) \le \psi(\mathcal{R}, n, \gamma, Q^{\text{com}}, x_1). \qedhere
\]
\end{proof}

\renewcommand{\algorithmcfname}{Game}
\begin{algorithm}[t]
Compute $Q^{\mathrm{B}},\alpha$ of $\mathcal{R}$\\
$y_1 \gets_{\mathcal{R}(x_1)} \mathbb{Y}$\\
$N \sim 1+\mathrm{Bin}(n-1,\alpha)$\\
$I\gets \emptyset$\\
\For{$i=2$ to $N$}{
    $c_i\sim \mathrm{Bern}(\frac{\gamma}{\alpha})$\\
    \If{$c_i=1$}{
        $y_i\gets_{Q^{\text{com}}} \mathbb{Y}$
    }
    \Else{
        $y_i\gets_{Q^{l}} \mathbb{Y}$\\
        $I\gets I \cup \{i\}$\\
    }
}
$\boldsymbol{y} \gets (y_1, y_2, \dots, y_N)$\\
$\sigma \gets_{\$} \mathcal{P}_N$ \tcp{Sample a random permutation}
$\boldsymbol{y}_\sigma \gets (y_{\sigma(1)}, y_{\sigma(2)}, \dots, y_{\sigma(N)})$\\
$g \gets \mathcal{A}(\boldsymbol{y}_\sigma,\sigma^{-1}(I))$\\
\textbf{return} $\sigma^{-1}(1) = g$
\caption{$\mathbf{GuessGame3}^{\mathcal{R},n,\gamma,Q^{\text{com}}}_{\mathcal{A}}(x_1)$}
\label{guess_game5}
\end{algorithm}

\begin{theorem}\label{theorem_mba}
Let $\mathcal{R}$ be a local randomizer in the shuffle DP model, and fix an input $x_1 \in \mathbb{X}$. Define
\[
M := \sup_{y} \frac{\Pr[\mathcal{R}(x_1) = y]}{\inf_{x \in \mathbb{X}} \Pr[\mathcal{R}(x) = y]}.
\]
Then the multiplicative Bayesian advantage satisfies:
\begin{align*}
\forall n \in \mathbb{N},\ \mathsf{Adv}_n^{\times}(\mathcal{R}, x_1) \le M, \quad
\text{and} \lim_{n \to \infty} \mathsf{Adv}_n^{\times}(\mathcal{R}, x_1) = M.
\end{align*}
\end{theorem}

The proof is provided in Appendix~\ref{appendix_3}. The proof of the upper bound follows a similar approach to that of Theorem~\ref{theorem_clone}. The tightness of this bound can be justified using results from the basic setting. Here, $M$ denotes the supremum of the likelihood ratios between $\mathcal{R}(x_1)$ and all other $\mathcal{R}(x)$ for $x \in \mathbb{X}$. Roughly speaking, we may assume that this maximum is attained at $x = x^*$. 
Now, let $x_2 = x_3 = \dots = x_n = x^*$. This particular instance of the shuffle DP setting reduces to a basic setting. By definition, we then have
\[
\mathsf{Adv}_n^{\times}(\mathcal{R}, x_1) \ge \beta_n(\mathcal{R}(x_1), \mathcal{R}(x^*)).
\]
By Theorem~\ref{theorem:asymptotic}, $\beta_n(\mathcal{R}(x_1), \mathcal{R}(x^*)) \to M$ as $n \to \infty$, and the result follows from the squeeze theorem.

On the one hand, Theorem~\ref{theorem_mba} implies Theorem~\ref{theorem_clone} because the differential privacy guarantee of $\mathcal{R}$ ensures that
\[
\sup_{y} \frac{\Pr[\mathcal{R}(x_1) = y]}{\inf_{x} \Pr[\mathcal{R}(x) = y]} \le e^{\varepsilon}.
\]
This also demonstrates that the bound in Theorem~\ref{theorem_clone} is tight in the worst case. On the other hand, the bound obtained from the \emph{blanket decomposition} is input-dependent and can yield sharper bounds for specific values of $x_1$. We illustrate this with the following example.

\begin{exmp}
Consider the Laplace mechanism defined over $[0,1]$:
\[
\mathcal{R}(x) = x + \mathrm{Lap}\left(\frac{1}{\varepsilon}\right), \quad x \in [0,1],
\]
where $\mathrm{Lap}(\delta)$ denotes the Laplace distribution with density function $\frac{1}{2\delta} e^{-|x|/\delta}$. A direct computation shows that the \emph{blanket distribution} of $\mathcal{R}$ is $\mathcal{R}(0.5)$, and the associated mixture coefficient is $\alpha = e^{-\varepsilon/2}$.

According to Theorem~\ref{theorem_mba}, we have the bound
\[
\mathsf{Adv}_n^{\times}(\mathcal{R}, x_1) \le \sup_{y} \frac{\Pr[\mathcal{R}(x_1) = y]}{\inf_{x} \Pr[\mathcal{R}(x) = y]} = e^{(|x_1 - 0.5| + 0.5)\varepsilon}.
\]
In particular, when $x_1 = 0.5$, the bound given by the blanket decomposition yields $\mathsf{Adv}_n^{\times}(\mathcal{R}, x_1) \le e^{\varepsilon/2}$, which is strictly tighter than the clone decomposition bound $\mathsf{Adv}_n^{\times}(\mathcal{R}, x_1) \le e^{\varepsilon}$.
\end{exmp}

In Theorems~\ref{therm_optimal} and \ref{theorem_mba}, the notation \(\mathcal{R}(x_i)\) is used. Here, we explain how these results naturally extend to the general case where \(x_i \sim V_i\). Following the discussion in Section~\ref{sec_definition_attack_dp}, it suffices to consider \(\mathcal{R}^*(\mathrm{V}_i)\). In fact, the blanket distribution of \(\mathcal{R}^*\) coincides with that of \(\mathcal{R}\), because all possible distributions of \(\mathcal{R}^*\) are convex combinations of \(\mathcal{R}\), so its maximal common component remains unchanged:
$$\forall y: \inf_{\mathrm{V}\in \mathcal{D}[\mathbb{X}]} \Pr[\mathcal{R}^*(\mathrm{V})=y]=\inf_{x\in \mathbb{X}}\Pr [\mathcal{R}(x)=y].$$
Therefore, in the case where \(x_i \sim V_i\), the blanket reduction remains optimal, with the only modification being to interpret \(\mathcal{R}(x_1)\) as \(\mathcal{R}^*(\mathrm{V}_1)\).
The interpretation for Theorem~\ref{theorem_mba} is analogous.

\begin{exmp}
Consider $\mathcal{R}$ as a $\ln(3)$-DP 2-RR mechanism. For $x_1\sim V_1=(p,1-p)$, we have 
\begin{align*}
&\mathcal{R}(x_1) \sim (0.25+0.5p,0.75-0.5p),\\
&\forall y\in \{1,2\},\ \inf_{x \in \mathbb{X}} \Pr[\mathcal{R}(x) = y]=0.25,\\
&M=\sup_{y} \frac{\Pr[\mathcal{R}(x_1) = y]}{\inf_{x \in \mathbb{X}} \Pr[\mathcal{R}(x) = y]} = \max \{1+2p,3-2p\}.
\end{align*}
According to Theorem~\ref{theorem_mba}, when $p=1$, we have $\mathsf{Adv}_n^{\times}(\mathcal{R}, x_1)\le3=e^{\epsilon}$. When $p=0.5$, $\mathsf{Adv}_n^{\times}(\mathcal{R}, x_1)\le 2 < e^{\epsilon}$.
\end{exmp}

\section{Conclusion}

In this paper, we presented the first systematic information-theoretic analysis of re-identification attacks in the shuffle model. We introduced a fundamental formulation, where one message is drawn from a distribution $P$ and the remaining $n-1$ messages from a distribution $Q$, and all are anonymized via shuffling. We derived an exact analytical expression for the success probability $\beta_n(P, Q)$ of the optimal Bayesian adversary, as well as its asymptotic behavior. We further established tight mutual bounds between the Bayesian advantage and the total variation distance $\Delta(P, Q)$, highlighting their equivalence up to a $1/n$ factor.

Extending beyond this basic setting, we analyzed re-identification attacks under general shuffle differential privacy protocols. We proposed a reduction framework that transforms the shuffle DP setting into a canonical instance, enabling analysis via known quantities $\beta_n(P, Q)$. Leveraging decomposition-based techniques, we demonstrated how the \emph{clone} and \emph{blanket} decompositions can be used to upper bound the adversary's success probability. Notably, we showed that the blanket decomposition not only yields the optimal privacy amplification bounds in prior literature but also provides the optimal upper bound for Bayesian re-identification attacks among all decomposition-based methods.

Our results establish new theoretical foundations for understanding anonymity leakage in the shuffle model and offer new insights into the security analysis of honeyword systems, shuffle DP mechanisms, and beyond. We hope this framework can inspire further research.

\bibliographystyle{IEEEtran}
\bibliography{bayes}

\appendix
\subsection{Proof of Theorem \ref{therm3}}\label{appendix_1}

\begin{proof}
Consider the case where $m - 1$ of the entries in $\{y_i\}_{i=2}^n$ have the same likelihood ratio as $y_1$. For the adversary to succeed, all $m$ of these entries—including $y_1$—must simultaneously attain the maximum likelihood ratio among all $n$ inputs. Since these $m$ entries are indistinguishable under the likelihood ranking, the attacker can do no better than guessing uniformly among them. Therefore, the success probability in this scenario is $1/m$.

We compute the total success probability:
\begin{align*}
&\beta_n(P, Q) \\
&= f(+\infty) + \sum_t \sum_{m=1}^n \frac{1}{m} f(t) \binom{n-1}{m-1} G^{n-m}(t^-) g^{m-1}(t) \\
&= f(+\infty) + \sum_t \sum_{m=1}^n \frac{1}{n} f(t) \binom{n}{m} G^{n-m}(t^-) g^{m-1}(t) \\
&= f(+\infty) + \sum_t \sum_{m=1}^n \frac{1}{n} \cdot t \cdot g(t) \binom{n}{m} G^{n-m}(t^-) g^{m-1}(t) \\
&= f(+\infty) + \sum_t \frac{1}{n} \cdot t \left[ \left( G(t^-) + g(t) \right)^n - G^n(t^-) \right] \\
&= f(+\infty) + \sum_t \frac{1}{n} \cdot t \left[ G^n(t) - G^n(t^-) \right],
\end{align*}
where we use the identity $\frac{1}{m} \binom{n-1}{m-1} = \frac{1}{n} \binom{n}{m}$, the binomial expansion, and the identity $f(t) = t \cdot g(t)$ (Property~\ref{prop0}).
\end{proof}

\subsection{Proof of Theorem \ref{crypto-game}}\label{appendix_2}
\begin{proof}

\textbf{Lower bound.}
We construct an explicit adversary $\mathcal{A}^*$: it samples a random index $r \in [n]$. If $P(y_r) \ge Q(y_r)$, the adversary outputs $r$; otherwise, it outputs a uniformly random index from $[n] \setminus \{r\}$.

Let $\sigma$ denote the shuffling permutation. If $\sigma^{-1}(1) = r$, then $y_r \sim P$, and the attack succeeds when $P(y_r) > Q(y_r)$, which occurs with probability $\sum_{x: P(x) > Q(x)} P(x)$. Otherwise, when $\sigma^{-1}(1) \ne r$, we have $y_r \sim Q$, and the adversary wins with probability $\frac{1}{n-1}$ if $P(y_r) < Q(y_r)$, which occurs with probability $\sum_{x: P(x) < Q(x)} Q(x)$.

Thus, the total success probability is:
\begin{align*}
&\beta_n(P, Q)\\
&\ge \mathbf{GuessGame}^{P,Q,n}_{\mathcal{A}^*}(1)\\
&= \frac{1}{n} \sum_{x: P(x) > Q(x)} P(x) + \frac{n-1}{n} \cdot \frac{1}{n-1} \sum_{x: P(x) < Q(x)} Q(x) \\
&= \frac{1}{n} \left[ \sum_{x: P(x) > Q(x)} P(x) + \sum_{x: P(x) < Q(x)} Q(x) \right] \\
&= \frac{1}{n} \sum_x \left( \frac{P(x) + Q(x)}{2} + \frac{|P(x) - Q(x)|}{2} \right) \\
&= \frac{1}{n} \left( 1 + \Delta(P, Q) \right),
\end{align*}
which implies:
\[
\mathsf{Adv}_n^+(P, Q) = \beta_n(P, Q) - \frac{1}{n} \ge \frac{1}{n} \cdot \Delta(P, Q).
\]

\medskip
\noindent\textbf{Upper bound.}
We first introduce two supporting lemmas.

\begin{lemma}\label{lemma:intG}
Let $P$ and $Q$ be two probability distributions, and let $G(t)$ denote the cumulative distribution function of the likelihood ratio $\frac{P(x)}{Q(x)}$ with respect to $x \sim Q$. Define $f(+\infty) := \sum_{x : Q(x) = 0} P(x)$.
For $ M \ge \sup\limits_{x : Q(x) > 0} \frac{P(x)}{Q(x)}$, we have
\[
\int_0^M G(t) \, \mathrm{d}t = M - 1 + f(+\infty).
\]
\end{lemma}

\begin{proof}
We apply integration by parts and use Property~\ref{prop0}, which states that $f(t) = t \cdot g(t)$, where $g$ is the density of $G$. Note that:
\begin{align*}
\int_0^M G(t) \, \mathrm{d}t 
&= t G(t) \big|_0^M - \int_0^M t \, \mathrm{d}G(t) \\
&= M \cdot G(M) - \int_0^M t \cdot g(t) \, \mathrm{d}t \\
&= M - \int_0^M f(t) \, \mathrm{d}t.
\end{align*}
Since $P$ is a probability distribution, the total mass of $f$ is:
\[
\int_0^M f(t) \, \mathrm{d}t = \sum_{x : Q(x) \ne 0} P(x) = 1 - f(+\infty),
\]
because $f(+\infty) = \sum_{x : Q(x) = 0} P(x)$. Substituting this in yields:
\[
\int_0^M G(t) \, \mathrm{d}t = M - (1 - f(+\infty)) = M - 1 + f(+\infty).
\]
\end{proof}

\begin{lemma}\label{lemma:tv_integral}
Let $P$ and $Q$ be two probability distributions. Then the total variation distance between $P$ and $Q$ can be expressed as:
\[
\Delta(P, Q) = \int_0^1 G(t) \, \mathrm{d}t,
\]
where $G(t)$ is defined as in Lemma~\ref{lemma:intG}.
\end{lemma}

\begin{proof}
We begin by writing:
\[
\Delta(P, Q) = \sum_{x : P(x) < Q(x)} [Q(x) - P(x)] = \sum_{t \le 1} \left[ Q(S_t) - P(S_t) \right],
\]
where $S_t := \{x \in \mathbb{X} \mid \frac{P(x)}{Q(x)} = t \}$.

Using the definition of the densities $f$ and $g$ over the likelihood ratio $t$, we can write:
\[
\Delta(P, Q) = \int_0^1 [g(t) - f(t)] \, \mathrm{d}t.
\]

By Property~\ref{prop0}, we know that $f(t) = t \cdot g(t)$, so:
\[
g(t) - f(t) = (1 - t) \cdot g(t).
\]

Thus:
\begin{align*}
\Delta(P, Q) &= \int_0^1 (1 - t) \cdot g(t) \, \mathrm{d}t \\
&= (1 - t) G(t) \big|_0^1 + \int_0^1 G(t) \, \mathrm{d}t \\
&= \int_0^1 G(t) \, \mathrm{d}t.
\end{align*}
In the last step, we used the fact that $(1 - t) G(t)$ vanishes at both $t = 0$ and $t = 1$.
\end{proof}

We begin with the expression from Theorem~\ref{theorem2}:
\[
\mathsf{Adv}_n^+(P, Q) = \frac{1}{n} \left( M - \int_0^M G^n(t) \, \mathrm{d}t \right) + f(+\infty) - \frac{1}{n}.
\]
This expression holds for any $M \ge \sup\limits_{x : Q(x) > 0} \frac{P(x)}{Q(x)}$. For the derivation that follows, we require $M > 1$. To satisfy this condition, we define
$M = \max\left\{2025,\ \sup\limits_{x : Q(x) > 0} \frac{P(x)}{Q(x)}\right\}$.

By Lemma~\ref{lemma:intG} and Lemma~\ref{lemma:tv_integral}, we have:
\[
\int_0^M G(t) \, \mathrm{d}t = M - 1 + f(+\infty), \quad \int_0^1 G(t) \, \mathrm{d}t = \Delta(P, Q).
\]
Subtracting these gives:
\[
\int_1^M G(t) \, \mathrm{d}t = M - 1 + f(+\infty) - \Delta(P, Q).
\]

To upper bound $\mathsf{Adv}_n^+(P, Q)$, we lower bound $\int_0^M G^n(t) \, \mathrm{d}t$. Observe:
\[
\int_0^M G^n(t) \, \mathrm{d}t \ge \int_1^M G^n(t) \, \mathrm{d}t.
\]

Applying Jensen’s inequality to the convex function $x^n$:
\begin{align*}
\int_1^M G^n(t) \, \mathrm{d}t 
&\ge (M - 1) \left( \frac{1}{M - 1} \int_1^M G(t) \, \mathrm{d}t \right)^n \\
&= (M - 1) \left( 1 + \frac{f(+\infty) - \Delta(P, Q)}{M - 1} \right)^n.
\end{align*}

Substituting into the formula for $\mathsf{Adv}_n^+(P, Q)$:
\begin{align*}
&\mathsf{Adv}_n^+(P, Q)\\
&\le \frac{1}{n} \left( M - 1 - \int_1^M G^n(t) \, \mathrm{d}t \right) + f(+\infty) \\
&\le \frac{1}{n} (M - 1) \left[ 1 - \left( 1 + \frac{f(+\infty) - \Delta(P, Q)}{M - 1} \right)^n \right] + f(+\infty).
\end{align*}

We now apply Bernoulli’s inequality:
\[
(1 + x)^n \ge 1 + nx, \quad \text{for } x \ge -1.
\]
Since $f(+\infty) - \Delta(P, Q)+M-1=\int_1^M G(t)\mathrm{d}t \ge0$, the value of $x := \frac{f(+\infty) - \Delta(P, Q)}{M - 1} \ge -1$, and so:
\[
\left( 1 + \frac{f(+\infty) - \Delta(P, Q)}{M - 1} \right)^n \ge 1 + n \cdot \frac{f(+\infty) - \Delta(P, Q)}{M - 1}.
\]

Substituting this into the bound:
\begin{align*}
&\mathsf{Adv}_n^+(P, Q)\\
&\le \frac{1}{n}(M - 1) \left[ 1 - \left(1 + n \cdot \frac{f(+\infty) - \Delta(P, Q)}{M - 1} \right) \right] + f(+\infty) \\
&= - (f(+\infty) - \Delta(P, Q)) + f(+\infty) = \Delta(P, Q). 
\end{align*}

\end{proof}
\subsection{Proof of Theorem \ref{theorem_mba}}\label{appendix_3}
\begin{proof}
Let \( T := \sup_{y} \frac{\mathcal{R}(x_1)(y)}{Q^{\mathrm{B}}(y)} \). By Theorem~\ref{theorem_compute_psi}, we have:
\begin{align*}
\psi(\mathcal{R},& n, \alpha, Q^{\mathrm{B}}, x_1)\\
&= \sum_{m=0}^{n-1} \binom{n-1}{m} \alpha^m (1 - \alpha)^{n-1 - m} \cdot \beta_{m+1}(\mathcal{R}(x_1), Q^{\mathrm{B}}) \\
&\le \sum_{m=0}^{n-1} \binom{n-1}{m} \alpha^m (1 - \alpha)^{n-1 - m} \cdot \frac{T}{m+1} \\
&= \frac{T}{\alpha n} \sum_{m=0}^{n-1} \binom{n}{m+1} \alpha^{m+1} (1 - \alpha)^{n-1 - m} \\
&= \frac{T}{\alpha n} \left[ 1 - (1 - \alpha)^n \right] \\
&\le \frac{T}{\alpha n}.
\end{align*}

Now observe that:
\[
\frac{T}{\alpha} = \sup_y \frac{\mathcal{R}(x_1)(y)}{\alpha Q^{\mathrm{B}}(y)} 
= \sup_y \frac{\Pr[\mathcal{R}(x_1) = y]}{\inf_x \Pr[\mathcal{R}(x) = y]}=M,
\]
which implies the upper bound:
\[
\mathsf{Adv}_n^{\times}(\mathcal{R},x_1) \le M.
\]

For the lower bound, let \( x^* \in \mathbb{X} \) be such that\footnote{If the $\arg\max$ does not exist, then there exists a sequence $(x^*_m)$ such that $$\lim_{m\to \infty}\sup_y \frac{\Pr[\mathcal{R}(x_1) = y]}{\Pr[\mathcal{R}(x^*_m) = y]}=M.$$ in which case the argument follows by considering $\mathsf{Adv}_n^{\times}(\mathcal{R}(x_1),\mathcal{R}(x^*_m))$.}:
\[
x^* = \arg\max_{x\in\mathbb{X}} \sup_y \frac{\Pr[\mathcal{R}(x_1) = y]}{\Pr[\mathcal{R}(x) = y]}.
\]
Now consider the input vector where \( x_2 = \cdots = x_n = x^* \). This reduces the setting to the basic comparison case, so:
\[
\mathsf{Adv}_n^{\times}(\mathcal{R},x_1) 
\ge \mathsf{Adv}_n^{\times}(\mathcal{R}(x_1), \mathcal{R}(x^*)).
\]
By Theorem~\ref{theorem:asymptotic}, we have:
\begin{align*}
\mathsf{Adv}_n^{\times}(\mathcal{R}(x_1), \mathcal{R}(x^*)) \to 
\sup_y \frac{\Pr[\mathcal{R}(x_1) = y]}{\Pr[\mathcal{R}(x^*) = y]} =M.
\end{align*}
Therefore, by the squeeze theorem, we conclude that
\[
\lim_{n \to \infty} \mathsf{Adv}_n^{\times}(\mathcal{R},x_1) = M.
\]
This completes the proof.
\end{proof}

\end{document}